\documentclass[journal]{IEEEtran} % double column

\normalsize

\ifCLASSINFOpdf
\else
\fi
\usepackage{setspace} % column

\usepackage{cite}
\usepackage{ragged2e}
\usepackage{amssymb}
\usepackage{paralist}
\usepackage{amsmath}
\usepackage{amsthm}
\usepackage{graphicx}
\usepackage{epstopdf}
\usepackage{color}
\usepackage{extarrows}
\usepackage{multicol}
\usepackage{caption}
\usepackage{url}
\usepackage{bm}
\usepackage{footnote}
\usepackage{diagbox}
\usepackage{multirow}
\usepackage{array}
\usepackage{diagbox}
\usepackage{cases}

\usepackage[linesnumbered,ruled,vlined]{algorithm2e}
\usepackage{algorithmicx}

\usepackage[justification=centering]{subcaption}
\usepackage{booktabs}

\newtheorem{Lem}{Lemma}
\newtheorem{Thm}{Theorem}

\allowdisplaybreaks
\PassOptionsToPackage{bookmarks=false}{hyperref}

\def\be{ \begin{eqnarray} }
\def\ee{ \end{eqnarray} }
\usepackage{soul}
\usepackage{stfloats}

\begin{document}

\title{Model-Assisted Learning for Adaptive Cooperative Perception of Connected Autonomous Vehicles}

\author{Kaige Qu, \IEEEmembership{Member, IEEE}, Weihua Zhuang, \IEEEmembership{Fellow, IEEE}, Qiang Ye, \IEEEmembership{Senior Member, IEEE}, \\Wen Wu, \IEEEmembership{Senior Member, IEEE}, and Xuemin Shen, \IEEEmembership{Fellow, IEEE}%
% \thanks{This work was supported in part by the Peng Cheng Laboratory Major Key Project under Grant PCL2021A09-B2 and by the Natural Science Foundation of China under Grant 62201311.}
\thanks{This work was supported by the Natural Sciences and Engineering Research Council (NSERC) of Canada.}
\thanks{Kaige Qu, Weihua Zhuang, and Xuemin (Sherman) Shen are with the Department of Electrical and Computer Engineering, University of Waterloo, Waterloo, ON N2L 3G1, Canada (emails: \{k2qu, wzhuang, sshen\}@uwaterloo.ca).}
% \thanks{Qiang Ye is with the Department of Electrical and Software Engineering, University of Calgary, Calgary, AB T2N 1N4, Canada (email: qiang.ye@ucalgary.ca).}
\thanks{Qiang Ye is with University of Calgary, Calgary, AB T2N 1N4, Canada (email: qiang.ye@ucalgary.ca).}
\thanks{Wen Wu is with Peng Cheng Laboratory, Shenzhen, Guangdong, China, 518055 (email: wuw02@pcl.ac.cn). He contributed to this study while working as a postdoctoral fellow at the University of Waterloo, Canada.}
}

\maketitle

\begin{abstract}

Cooperative perception (CP) is a key technology to facilitate consistent and accurate situational awareness for connected and autonomous vehicles (CAVs). To tackle the network resource inefficiency issue in traditional broadcast-based CP, unicast-based CP has been proposed to associate CAV pairs for cooperative perception via vehicle-to-vehicle transmission. In this paper, we investigate unicast-based CP among CAV pairs. With the consideration of dynamic perception workloads and channel conditions due to vehicle mobility and dynamic radio resource availability, we propose an adaptive cooperative perception scheme for CAV pairs in a mixed-traffic autonomous driving scenario with both CAVs and human-driven vehicles. We aim to determine when to switch between cooperative perception and stand-alone perception for each CAV pair, and allocate communication and computing resources to cooperative CAV pairs for maximizing the computing efficiency gain under perception task delay requirements. A model-assisted multi-agent reinforcement learning (MARL) solution is developed, which integrates MARL for an adaptive CAV cooperation decision and an optimization model for communication and computing resource allocation. Simulation results demonstrate the effectiveness of the proposed scheme in achieving high computing efficiency gain, as compared with benchmark schemes.

\end{abstract}


\IEEEpeerreviewmaketitle

\section{Introduction}
\label{sec:Introduction}

The advances in sensing, artificial intelligence (AI), and vehicles-to-everything (V2X) communication technologies have paved the way for autonomous driving, leading to a potential paradigm shift in future transportation systems towards improved road safety and traffic efficiency~\cite{zhuang2019sdn,shen2021holistic,hui2022collaboration}. 
Reliable and real-time environment perception is a key component in autonomous driving that facilitates the connected and autonomous vehicles (CAVs) to accurately and continuously perceive the surrounding objects, such as traffic participants, by using on-board cameras, light detection and ranging (LiDAR) sensors, and radar sensors~\cite{wang2018networking,zhang2019mobile}.  
To enhance the perception reliability in terms of both coverage and accuracy, \emph{cooperative perception} (CP) has been proposed to enable the sensory information sharing among CAVs by leveraging V2X communication, as a supplement to the \emph{stand-alone perception} (SP) by individual CAVs based on their own viewpoints~\cite{zheng2022confidence,jia2022online,jia2023mass,abdel2021vehicular,xiao2022perception,sun2022user,lin2022low}.  
In case of unreliable network connectivity or network congestion due to limited network resources, CAVs can switch back to the default SP mode~\cite{zhang2021emp}.

According to the type of shared sensory information, there are three CP levels, including raw level, feature level, and decision level.  
In the raw-level CP, complete~\cite{jia2022online,chen2019cooper} or partial raw data~\cite{qiu2021autocast,zhang2021emp} are shared among CAVs, which preserves the most fine-grained environmental information and leads to the highest perception performance gain at the cost of huge communication overhead due to the large data volume. 
The decision-level CP integrates lightweight perception results of individual CAVs, which is communication-efficient but with limited perception performance gain~\cite{xiao2022perception}.  
The feature-level CP, which has gained significant attention in the computer vision field, can balance between communication overhead and perception performance gain~\cite{chen2019f,wang2020v2vnet}.  
Research studies on feature-level CP have focused on the design of AI-based fusion schemes, but the underlying communication scheme is usually simple, e.g., via broadcast-based vehicle-to-vehicle (V2V) communication. Specifically, each CAV fuses all the received feature data broadcast from neighboring CAVs with its own data and processes the fused data for inference~\cite{chen2019f,wang2020v2vnet}.  
Although the feature data are compressed from the raw data, the data size is still large, e.g., in the scale of Mbits. 
Hence, the broadcast-based CP schemes are communication inefficient especially in dense-traffic scenarios and even not applicable when the available transmission resources are limited.   
Moreover, due to individual computation at each CAV, the overall computation demand is intensive, which is roughly proportional to the number of CAVs and the data volume processed at each CAV~\cite{zhang2021emp}. 
The communication and computation in such broadcast-based CP schemes lead to large network resource consumption for satisfying the stringent delay requirement of  real-time perception tasks.

Recently, there are some studies on the deployment of CP schemes in a practical network environment by considering the limited V2X communication bandwidth and on-board computing resources~\cite{jia2022online,jia2023mass,xiao2022perception,abdel2021vehicular}. 
As nearby CAVs collect sensing data of common objects in a shared environment from diverse viewpoints, adding sensing data from more CAVs for data fusion potentially improves the perception performance with a diminishing marginal gain, at the cost of almost linearly increasing network resources.  
For resource efficiency, unicast-based CP schemes have been studied in~\cite{jia2022online,jia2023mass,abdel2021vehicular}, 
which deal with the association of CAV pairs that perform the cooperative perception via unicast-based V2V communication. 
Two CAVs with complementary or enhancing sensory information usually provide higher perception performance gain through cooperation and tend to be associated, due to more even spatial distribution or higher intensity of fused sensing data.    
In comparison with the broadcast-based counterparts, unicast-based CP schemes significantly improve the network resource efficiency without a remarkable compromise on the perception performance through proper association of CAV pairs.

In this work, we investigate unicast-based feature-level CP among CAVs. 
Different from the existing works on CAV pair association to improve the perception performance gain, we investigate how to support CP for predetermined CAV pairs in a complex and dynamic network environment with high resource efficiency. 
Specifically, we consider a practical mixed-traffic autonomous driving scenario where a cluster of CAVs and human-driven vehicles (HDVs) traverse a road segment with intermittent road-side-unit (RSU) coverage due to the high RSU deployment cost. 
Each CAV pair works in either the SP mode by default or the unicast-based feature-level CP mode by selection.  
Considering the radio resource sharing among vehicles, the radio resource availability for CAVs dynamically changes over time. 
Due to vehicle mobility, the perception workloads and channel conditions for CAVs are dynamic in different perception task periods. 
In such a network scenario, 
it is challenging to constantly support all CAV pairs to work in the CP mode with delay satisfaction. 

To accommodate the network dynamics, we propose an adaptive cooperative perception scheme, which facilitates a dynamic selection of CAV pairs for cooperative perception. 
The selected CAV pairs are referred to as \emph{cooperative CAV pairs}, and the non-selected CAV pairs work in the SP mode by default.  
For each cooperative CAV pair, we dynamically allocate a fraction of available radio resources to support the feature data transmission, and adjust the CPU frequency at the CAVs on demand, to balance between the transmission and computation delays under the network dynamics, while satisfying a perception delay requirement. 
For the joint adaptive CAV cooperation and resource allocation, there is a trade-off between a total computing efficiency gain and a total cost for dynamically switching between the SP and CP modes for the CAV pairs. 
Specifically, for a CAV pair, the total \emph{computing demand} is significantly reduced in the CP mode, 
by performing the data fusion and inference at one CAV and allowing the computation result sharing within the CAV pair. 
However, due to the on-demand CPU frequency allocation, the \emph{CPU frequency} in the CP mode can be occasionally higher than that in the SP mode. 
As the CAVs work in the SP mode by default, we characterize the computing efficiency gain of a CAV pair as the reduced amount of computing energy consumption in comparison with that in the SP mode, which depends on both computing demand and CPU frequency. 
We focus on increasing the total computing efficiency gain while reducing the total switching cost, via proper selection of cooperative CAV pairs and optimal resource allocation.  
The main contributions of this paper are summarized as follows. 
\begin{itemize}

    \item We propose an adaptive cooperative perception scheme for CAV pairs in a moving mixed-traffic vehicle cluster, which allows each CAV pair to dynamically switch between the SP and CP modes over different perception task periods, to adapt to the network dynamics; 

    \item 
    We formulate a joint adaptive CAV cooperation and resource allocation problem, which can be decoupled to an adaptive CAV cooperation subproblem in the long run and a series of instantaneous resource allocation subproblems in each perception task period, to maximize the total computing efficiency gain with minimum switching cost, while satisfying the perception delay requirement;    

    \item We propose a model-assisted multi-agent reinforcement learning (MARL) solution, where MARL is used to learn the adaptive cooperation decisions among CAV pairs, and a model-based solution is used for resource allocation given each cooperation decision. 

\end{itemize}

The remainder of this paper is organized as follows. The system model is presented in Section~\ref{sec:System Model}, with a performance analysis included in Section~\ref{sec:2D Performance Region Analysis}. The joint adaptive CAV cooperation and resource allocation problem is formulated in Section~\ref{sec:problem}, with a model-assisted MARL solution presented in Section~\ref{sec:solution}. Simulation results are provided in Section~\ref{sec:Simulation Results}, and conclusions are drawn in Section~\ref{Conclusion}.

\section{System Model}
\label{sec:System Model}

% double
\begin{figure}
\centering
{
\includegraphics[width=1\linewidth]{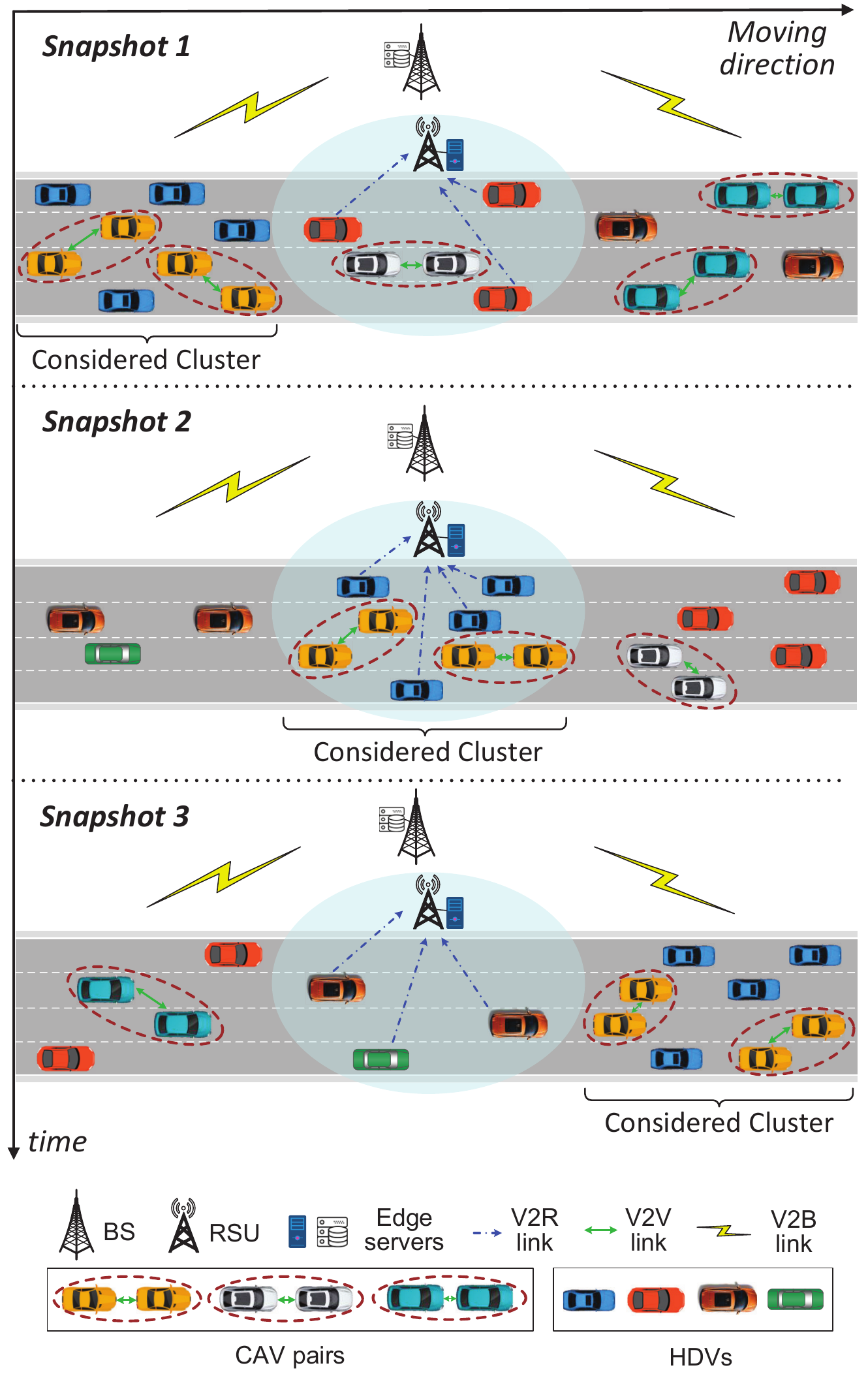}} 
\caption{A mixed-traffic autonomous driving scenario.}\label{fig:JP7-Scenario_with_BS}
\end{figure}

\subsection{Mixed-Traffic Autonomous Driving Scenario}

We consider a vehicle cluster including $M$ HDVs and $K$ CAV pairs in set $\mathcal{K}$, 
moving on a multi-lane unidirectional road under a consistent base station (BS) coverage and an intermittent RSU coverage. 
The CAV pairs are predetermined by using existing CAV pair association algorithms~\cite{jia2023mass,abdel2021vehicular}. 
a cluster head is selected among the vehicles based on existing vehicle clustering algorithms, to coordinate the communication, computing, and sensing in the vehicle cluster~\cite{wang2018networking}.  
Both BS and RSU provide the edge computing capability, facilitated by co-located edge servers. 
Fig.~\ref{fig:JP7-Scenario_with_BS} illustrates three snapshots for a vehicle cluster moving through an RSU's coverage area. Initially, all the vehicles in the cluster have no access to the RSU but the leading vehicle is about to move into the RSU coverage. Then, the vehicles gradually move into and later leave the RSU coverage. 
Consider a time-slotted system, in which the time slots are indexed by integer $n$. 

Each CAV initiates a perception task in each time slot to identify the surrounding objects, whose results are essential to supporting autonomous driving applications, such as path planning and maneuver control. 
Each CAV pair $k\in\mathcal{K}$ consists of one transmitter CAV and one receiver CAV, and both CAVs share a similar view but from different angles.  
Each CAV pair works in the SP mode by default and in the CP mode by selection. 
Let $\boldsymbol{x}(n)=\left\{ x_k(n), \forall k\in\mathcal{K}\right\}$ be a binary perception mode selection decision vector for all CAV pairs  (also referred to as cooperation decisions for brevity)  at time slot $n$, with $x_k(n)=1$ indicating the CP mode and $x_k(n)=0$ indicating the SP mode for CAV pair $k$.
A CAV pair in the CP mode is referred to as a cooperative CAV pair. 
Let $\mathcal{K}_C(n)$ denote a set of cooperative CAV pairs at time slot $n$, with $\mathcal{K}_C(n) \subset \mathcal{K}$. 

Each HDV occasionally requests an infotainment service, e.g., mobile virtual reality, which is throughput sensitive and computation intensive. For energy and delay efficiency, the computation tasks at an HDV can be offloaded to a more powerful edge server either at a BS or at an RSU, and the data transmission is supported by either vehicle-to-BS (V2B) or vehicle-to-RSU (V2R) communications~\cite{wu2020dynamic}.

\subsection{Perception Task Model}
\label{sec:Environment Perception Task Model}

For environment perception, lightweight data pre-processing algorithms are used to slice the raw sensing data into \emph{object partitions} each containing one object of interest and \emph{background partitions} that contain only background information~\cite{zhang2021elf}. 
An object tracking algorithm associates the object partitions with existing objects in a maintained object tracking list, by comparing the identified and predicted object locations~\cite{wang2022vabus}.  
Only the new objects and the objects with reduced tracking accuracy are further processed by a deep neural network (DNN) for classification~\cite{yang2022flexpatch}. 
For CAV pair $k$, let $W_k(n)$, $W_k^{\mathsf{T}}(n)$, and $W_k^{\mathsf{R}}(n)$ denote the number of objects that require DNN model processing in the overlapping sensing range of both CAVs, in the non-overlapping sensing range of the transmitter CAV, and in the non-overlapping sensing range of the receiver CAV, respectively, at time slot $n$. We refer to the objects in the overlapping and non-overlapping sensing ranges as shared and individual objects respectively. Then, $W_k(n)$ is also referred to as shared workload, and $W_k^{\mathsf{T}}(n)$, $W_k^{\mathsf{R}}(n)$ are referred to as individual workloads, for CAV pair $k$.

At time slot $n$, the object classification tasks for all the objects of CAV pair $k$ should be completed within a time duration of $\Delta$, which is typically smaller than the time slot length, with the consideration of 1) the raw data pre-processing and object tracking procedures before the generation of object classification tasks, and 2) the response time for autonomous driving applications based on the object classification results.

%double
\begin{figure}
\centering
{
\includegraphics[width=0.7\linewidth]{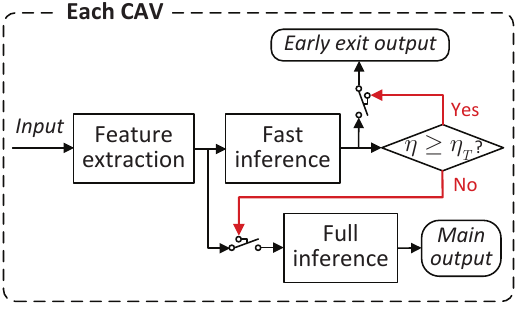}} 
\caption{Object classification by using a default DNN model.}\label{fig:Default-DNN-model}
\end{figure}

% double 
\begin{figure}
\centering
{
\includegraphics[width=0.9\linewidth]{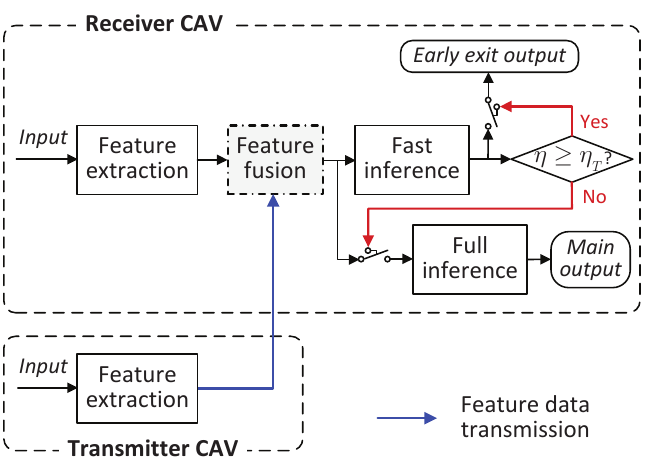}} 
\caption{Object classification by using a feature-fusion DNN model.}\label{fig:Data-fusion-DNN-model}
\end{figure}

\subsubsection{DNN models}

To support the object classification at each CAV pair, we consider a \emph{default DNN model} deployed at both the transmitter and receiver CAVs and a \emph{feature-fusion DNN model} partitioned between them, both employing an early-exit DNN architecture, as illustrated in Fig.~\ref{fig:Default-DNN-model} and Fig.~\ref{fig:Data-fusion-DNN-model} respectively~\cite{teerapittayanon2017distributed,10137743,li2019edge,liu2023resource}. 
As the objects can be processed independently, we consider that both DNN models operate on a per-object basis~\cite{zhang2021elf,wang2023real}. 
Specifically, in the default DNN model, a feature extraction module, mainly composed of convolution (\texttt{CONV}) layers, first generates compressed feature data based on an input of object sensing data. Then, the feature data are further processed by a fast inference module, composed of both \texttt{CONV} and fully-connected (\texttt{FC}) layers, to generate a DNN inference result, referred to as a fast inference result. 
Letting $Z$ denote the total number of object classes, a DNN inference result is a $Z$-dimension estimated class probability vector, where the classification performance is measured by confidence level defined as one minus normalized entropy~\cite{teerapittayanon2017distributed,10137743}. 
A higher confidence level indicates a less uncertain estimation and implies a higher accuracy~\cite{10137743}.
Let $\eta$ denote the confidence level of a fast inference result. If $\eta$ reaches a predetermined threshold, $\eta_T$, an object classification result is obtained at an early exit output.
Otherwise, a full inference module, composed of deeper \texttt{CONV} and \texttt{FC} layers, is triggered to re-process the feature data and generate another DNN inference result, referred to as a full inference result, at a main output.
Let $\rho\in(0,1)$ be the early exit probability for the default DNN model, representing the probability that an object classification result is obtained at the early exit output. 

For a shared object of a CAV pair, the transmitter and receiver CAVs have object sensing data from different viewpoints, implying a potential confidence level gain from data fusion. Such an object can be processed by using the feature-fusion DNN model. 
Specifically, both CAVs process their own object sensing data and extract features based on a feature extraction module. 
Then, the feature data of the transmitter CAV are transmitted via V2V communication to the receiver CAV, where the feature data from both CAVs are fused and processed by the fast inference and selective full inference modules.
The object classification result is obtained at either an early exit output with probability $\tilde{\rho}\in(0,1)$, or a main output with probability $1-\tilde{\rho}$, at the receiver CAV, which is then sent back to the transmitter CAV.  
Typically, we have $\tilde{\rho}>\rho$, as more fast inference results can satisfy the confidence level requirement due to the confidence level gain from data fusion. 

For the DNN models, let $\delta_1$, $\delta_2$, $\delta_3$ and $\delta_4$ be the computing demand (in CPU cycles) for feature extraction, feature fusion, fast inference, and full inference. 
Typically, we have $\delta_2\ll \min\{\delta_3,\delta_4\}$, as the feature fusion module can be implemented by simple operations such as concatenation, maxout, and average operations~\cite{teerapittayanon2017distributed,chen2019f} or lightweight attention schemes~\cite{wang2020v2vnet}. %li2023learning,
The average computing demand for processing one object by the default and feature-fusion DNN models, denoted by $\delta$ and $\tilde{\delta}$ respectively, are given by
% double
\begin{align} 
  \delta &= \delta_1 + \delta_3 + (1-\rho)\delta_4 \\
    \tilde{\delta} &= 2\delta_1 + \delta_2 + \delta_3 + (1-\tilde{\rho})\delta_4.
    \label{eq-compu-demand}
\end{align}

\subsubsection{Cooperative Perception Mode}

For a CAV pair in the CP mode, the shared objects are collaboratively processed by using the feature-fusion DNN model, while the individual objects are independently processed at each CAV by using the default DNN model. 
Consider two CPU cores at each CAV, for processing the shared and individual objects separately using different DNN models. For each CPU core, the dynamic voltage and frequency scaling (DVFS) technique is used to allow on-demand CPU frequency scaling, to support the dynamic perception workloads~\cite{jia2022online}.  

\subsubsection{Stand-Alone Perception Mode}

In the default SP mode, both CAVs in a CAV pair independently perform the object classification tasks. Both the shared and individual objects are processed by using the default DNN model. For ease of analysis, we assume that the shared and individual objects are processed at separate CPU cores, as in the CP mode.  

\subsection{Computing Model}

\begin{table}[t]
\small
\centering
\caption{\scshape{CPU Frequency Configuration for a CAV Pair}}
\label{Table:CPUs}
\begin{tabular}{ |c|c|c|c|c| } 
\hline
\textsl{CPU Core} & \multicolumn{2}{c|}{ Core 1 } & \multicolumn{2}{c|}{  Core 2 }\\ \hline
\textsl{Object Type} & \multicolumn{2}{c|}{ Shared Objects  } & \multicolumn{2}{c|}{  Individual Objects }\\ \hline
\textsl{Perception Mode} & CP & SP & CP & SP \\ \hline
\textsl{CPU Frequency} & \multirow{2}{*}{$f_k(n)$} & \multirow{2}{*}{$f^{\mathsf{D}}_k(n)$} & \multirow{2}{*}{$f_k^{\mathsf{T}}(n)$} & \multirow{2}{*}{$f_k^{\mathsf{T}}(n)$}\\ 
\textsl{at Transmitter CAV} &  &  &  & \\ \hline
\textsl{CPU Frequency} & \multirow{2}{*}{$f_k(n)$} & \multirow{2}{*}{$f^{\mathsf{D}}_k(n)$} & \multirow{2}{*}{$f_k^{\mathsf{R}}(n)$} & \multirow{2}{*}{$f_k^{\mathsf{R}}(n)$}\\ 
\textsl{at Receiver CAV} &  &  &  & \\ \hline
\end{tabular}
\end{table}

For CAV pair $k$, let $\delta_k^{C}(n)$, $\delta_k^{S}(n)$, $\delta_k^{T}(n)$, and $\delta_k^{R}(n)$ be the average total computing demand for processing the shared objects in the CP mode, the shared objects in the SP mode, the transmitter CAV's individual objects, and the receiver CAV's individual objects, at time slot $n$, given by
\begin{align}  
    \delta_k^{C}(n)  =  \tilde{\delta} W_k(n), \quad \delta_k^{S}(n)  =  2\delta W_k(n) 
\end{align}
\begin{align}  
    \delta_k^{T}(n)  =  \delta W_k^{\mathsf{T}}(n), \quad \delta_k^{R}(n)  = \delta W_k^{\mathsf{R}}(n). 
\end{align}
% double
There is a positive total computing demand reduction for CAV pair $k$ through cooperation which increases proportionally to shared workload $W_k(n)$, given by $(2\delta - \tilde{\delta}) W_k(n)   =  \left[\delta_3 + \left( 1 + \tilde{\rho} - 2\rho   \right) \delta_4 -\delta_2\right] W_k(n)>0$.

For CAV pair $k$, let $f_k(n)$, $f_k^{\mathsf{T}}(n)$, and $f_k^{\mathsf{R}}(n)$ denote the CPU frequencies (in Hz or cycle/s) for processing the shared objects at both CAVs, the individual objects at transmitter CAV, and the individual objects at receiver CAV, respectively, at time slot $n$. We have $f_k(n) = f^{\mathsf{D}}_k(n)$ if $x_n(k)=0$, where $f^{\mathsf{D}}_k(n)$ is the CPU frequency for processing the shared objects at both CAVs in the default SP mode. 
Table~\ref{Table:CPUs} summarizes the relationships among CPU cores, object types, perception modes, and CPU frequencies for a CAV pair.  
As the shared and individual objects can be processed in parallel at the separate CPU cores, we have $f^{\mathsf{D}}_k(n) = \frac{  \delta W_k(n)}{\Delta}$, $f_k^{\mathsf{T}}(n)=\frac{\delta W_k^{\mathsf{T}}(n)}{\Delta}$, and $f_k^{\mathsf{R}}(n)=\frac{\delta W_k^{\mathsf{R}}(n)}{\Delta}$
to ensure that the default DNN model processing can be finished within delay bound $\Delta$ for the shared objects in the SP mode, and for the individual objects in either SP or CP mode, without CPU frequency over-provisioning. 
All CPU frequencies are upper limited by a maximum CPU frequency, $f_{\mathsf{M}}$, supported by DVFS, leading to an upper limit, $W_{\mathsf{M}} = \frac{f_{\mathsf{M}} \Delta}{\delta}$, for $W_k(n)$, $W_k^{\mathsf{T}}(n)$, and $W_k^{\mathsf{R}}(n)$.

The total computing energy consumption by all CPU cores of CAV pair $k$ at time slot $n$, denoted by $e_k(n)$, is given by
% double
\begin{align} 
e_k(n) = & ~\kappa \left[ f_k(n)^2 \delta_k^{C}(n)x_k(n) + f_k^{\mathsf{D}}(n)^2 \delta_k^{S}(n)\left(1-x_k(n)\right) \right.\nonumber\\
& \left. + f_k^{\mathsf{T}}(n)^2 \delta_k^{T}(n) + f_k^{\mathsf{R}}(n)^2 \delta_k^{R}(n) \right] 
\label{eq-energy} 
\end{align}
where $\kappa$ is the energy efficiency coefficient of a CPU core~\cite{10137743}. 
From \eqref{eq-energy}, we see that only the portion of computing energy for processing the shared objects depends on cooperation decision $x_k(n)$. 
The computing energy is a comprehensive metric that integrates both CPU frequency and computing demand. 
We characterize the computing efficiency gain of CAV pair $k$ at time slot $n$, denoted by $G_k(n)$, as the reduced amount of computing energy in comparison with that in the default SP mode.  
We have $G_k(n)\equiv0$ for $k \notin\mathcal{K}_C(n)$. For cooperative CAV pair $k\in\mathcal{K}_C(n)$, we have 
\begin{align} 
    &G_k(n) = \kappa f_k^{\mathsf{D}}(n)^2 \delta_k^{S}(n) - \kappa f_k(n)^2 \delta_k^{C}(n) \nonumber\\
    &= 2\kappa  \delta f^{\mathsf{D}}_k(n)^2 W_k(n) - \kappa \tilde{\delta} f_k(n)^2  W_k(n), ~\forall k\in\mathcal{K}_C(n)
    \label{eq-reduced-energy} 
\end{align}
as a decreasing function of $f_k(n)>0$. 
Here, $G_k(n)$ is independent of the individual workloads, as only shared objects are processed differently between the SP and CP modes.

\subsection{Communication Model}

Consider a radio resource pool with total bandwidth $B$ for V2X sidelink communication, which is shared between the V2R transmission from HDVs and the V2V transmission for cooperative CAV pairs, and is non-overlapping with that for V2B communication. 
The radio resource sharing between CAVs and HDVs occur only in the RSU coverage. 
Consider a transmission priority for HDVs, as the CAVs can work in the SP mode without radio resource usage by default. % in the RSU coverage
Orthogonal frequency division multiplexing (OFDM) based transmission schemes are employed for the V2X sidelink communication. % in the vehicle cluster. 

Let $B(n)$ be the time-varying available radio spectrum bandwidth for CAVs, with the consideration of dynamic background radio resource usage by HDVs. 
Let $\boldsymbol{\beta}(n)=\left\{ \beta_k(n), \forall k\in\mathcal{K}\right\}$ be a radio resource allocation decision vector for all CAV pairs at time slot $n$, with $\beta_k(n)$ denoting the fraction of available radio resources allocated to CAV pair $k$ for supporting the feature data transmission at time slot $n$.
The average transmission rate for CAV pair $k$ at time slot $n$ is 
\begin{align} 
    R_k(n) =  \beta_k(n) B(n) \log_2 \left( 1+\frac{ p_k g_k(n) }{ \sigma^2} \right), \quad \forall k\in\mathcal{K}
    \label{eq-achieved-rate}
\end{align}
where $p_k$ is the transmit power of CAV pair $k$, $g_k(n)$ is the channel power gain between both CAVs in CAV pair $k$ at time slot $n$, and $\sigma^2$ represents the received noise power. 
Due to high vehicle mobility, we consider only the large-scale channel conditions, specifically the path loss, for the CAV pairs.

\subsection{Delay Model}

Under the assumption of $\max\{W_k(n), W_k^{\mathsf{T}}(n), W_k^{\mathsf{R}}(n)\} \leq W_{\mathsf{M}}$, the CPU frequencies for processing the shared objects in the SP mode and for processing the individual objects in both SP and CP modes can be feasibly scaled up/down to ensure delay satisfaction. 
Here, we focus on the delay performance of the shared objects in the CP mode. 
Let $w$ denote the feature data size in the unit of bit. 
If CAV pair $k$ works in the CP mode, the average object classification delay for each shared object, denoted by $d_k(n)$, depends on the feature-fusion DNN model.  
The delay is composed of 
feature extraction delay $\frac{\delta_1}{f_k(n)}$, feature data transmission delay $\frac{w}{R_k(n)}$, 
feature fusion delay $\frac{\delta_2}{f_k(n)}$, 
and the average inference delay, $\frac{ \delta_3 + \left(1-\tilde{\rho}\right)\delta_4}{f_k(n)}$. 
The delay for sending the classification results is negligible due to the small data size.
For delay satisfaction, $d_k(n)$ should not exceed a per-object delay budget, $\frac{\Delta}{W_k(n)}$, given by
\begin{align} 
    d_k(n) =\frac{w}{R_k(n)} + \frac{ \hat{\delta}}{f_k(n)} \leq \frac{\Delta}{W_k(n)}, \quad \forall k\in\mathcal{K}_C(n)
    \label{eq-ave-delay-budget}
\end{align}
where $\hat{\delta}=\delta_1 + \delta_2 + \delta_3 + \left(1-\tilde{\rho}\right)\delta_4$ is a constant.

\subsection{Generalization}

For computing efficiency gain and perception accuracy enhancement, nearby CAVs can be grouped for cooperative perception by using a $Y$-input feature-fusion DNN model, where $Y$ is a general group size. 
At a given vehicle density, the content similarity within a group tends to reduce as $Y$ increases, due to a lower average percentage of shared objects. 
As only the shared objects can be collaboratively processed by using the feature-fusion DNN model, the reduced content similarity may gradually compromise the total computing efficiency gain and the average perception accuracy enhancement as $Y$ increases. 
Additionally, a larger group size increases the accuracy of shared objects with a diminishing marginal gain. 
However, there is a higher overall communication cost for supporting the feature data transmission of more transmitter CAVs in a larger group. 
For simplicity, we consider $Y=2$ and pair the CAVs based on existing works~\cite{jia2022online,jia2023mass,abdel2021vehicular}. How to select the best group size, $Y$, and how to optimally group the CAVs given $Y$ remain to be investigated in our future work.

\section{2D Performance Region Analysis}
\label{sec:2D Performance Region Analysis}

For problem formulation, we analyze the performance of an arbitrary cooperative CAV pair, $k\in\mathcal{K}_C(n)$, at time slot $n$, under different transmission rates and CPU frequencies for supporting the classification of shared objects. 
The condition for non-negative computing efficiency gain via cooperation, i.e., $G_k(n) \geq 0$, is a CPU frequency requirement, given by
\begin{align} 
    f_k(n) \leq f^{\mathsf{P}}_k(n) = \sqrt{\frac{2\delta}{\tilde{\delta}}}f^{\mathsf{D}}_k(n) = \sqrt{\frac{2\delta^3}{\tilde{\delta}}} \frac{W_k(n)}{\Delta}
\end{align}
where $f^{\mathsf{P}}_k(n)$ corresponds to zero computing efficiency gain and increases proportionally to shared workload $W_k(n)$.
We have $f^{\mathsf{P}}_k(n) > f^{\mathsf{D}}_k(n)$, as $\sqrt{\frac{2\delta}{\tilde{\delta}}}>1$. 
Under assumption $W_k(n)\leq W_{\mathsf{M}}$, there are two cases for the relationship among $f^{\mathsf{D}}_k(n)$, $f^{\mathsf{P}}_k(n)$ and $f_{\mathsf{M}}$, depending on shared workload $W_k(n)$:

\begin{itemize}

    \item \emph{Low shared workload:} For $W_k(n) \leq \sqrt{\frac{\tilde{\delta}}{2\delta}} \frac{f_{\mathsf{M}}\Delta }{\delta}  = \sqrt{\frac{\tilde{\delta}}{2\delta}}W_{\mathsf{M}} $, we have $f^{\mathsf{D}}_k(n) < f^{\mathsf{P}}_k(n) \leq f_{\mathsf{M}}$. 
    For a feasible CPU frequency scale-up from $f^{\mathsf{D}}_k(n)$ to $f_{\mathsf{M}}$, computing efficiency gain $G_k(n)$ transits from positive to negative, with a zero value at $f_k(n)=f^{\mathsf{P}}_k(n)$; 
    
    \item \emph{High shared workload:} For $\sqrt{\frac{\tilde{\delta}}{2\delta}}W_{\mathsf{M}}  < W_k(n) \leq W_{\mathsf{M}}$, we have $f^{\mathsf{D}}_k(n) \leq f_{\mathsf{M}} < f^{\mathsf{P}}_k(n)$. As $f_k(n)$ scales up from $f^{\mathsf{D}}_k(n)$ to $f_{\mathsf{M}}$, $G_k(n)$ decreases but remains positive. 

\end{itemize}

% double
\begin{figure}
    \centering
    \begin{subfigure}[b]{0.24\textwidth}
           \includegraphics[width=1\linewidth]{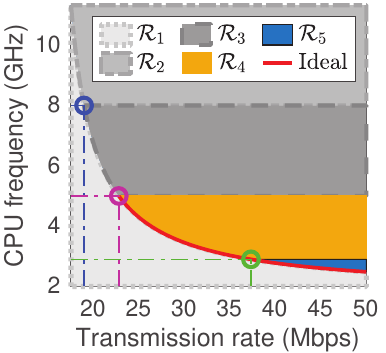} 
        \caption{\small{Low workload}}\label{fig:tradeoff}
    \end{subfigure}
    % ~ 
    \begin{subfigure}[b]{0.24\textwidth}
    \includegraphics[width=1\linewidth]{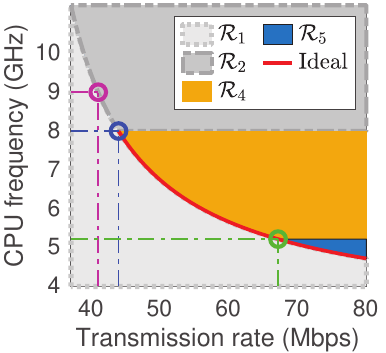} 
    \caption{\small{High workload}}\label{fig:tradeoff_2}
    \end{subfigure}
    \caption{Examples of 2D performance regions for cooperative CAV pair $k$ at a different shared workload, where $\left[R_{\mathsf{M}},f_{\mathsf{M}}\right]$, $\left[R^{\mathsf{P}}_k(n),f^{\mathsf{P}}_k(n)\right]$, and $\left[R^{\mathsf{D}}_k(n),f^{\mathsf{D}}_k(n)\right]$ are indicated by blue, pink, and green circles.}\label{fig:tradeoff_together}
\end{figure}

% double
\begin{table*}[t]
% \scriptsize
\small
\centering
\caption{\scshape{Summary of Different Performance Regions of a Cooperative CAV Pair}}
\label{Table:performance-regions}
\begin{tabular}{c| c | c |c |c } %p{0.79\columnwidth}
\hline\noalign{\vskip 0.3mm}\hline\noalign{\smallskip}
\textbf{Region} & \textbf{Shared Workload} & \textbf{Delay Performance} & \textbf{CPU Frequency Scaling Range} & \textbf{Computing Efficiency Gain}  \\
\noalign{\smallskip}\hline\noalign{\smallskip}
$\mathcal{R}_1$ & Low/High & Violation & N/A & N/A    \tabularnewline \noalign{\smallskip}\hline\noalign{\smallskip}
$\mathcal{R}_2$ & Low/High & Satisfaction & Infeasible scale-up: $f_k(n) > f_{\mathsf{M}}$ & N/A    \tabularnewline \noalign{\smallskip}\hline\noalign{\smallskip}
$\mathcal{R}_3$ & Low & Satisfaction & Feasible scale-up: $f^{\mathsf{P}}_k(n) < f_k(n) \leq f_{\mathsf{M}}$ & Negative    \tabularnewline \noalign{\smallskip}\hline\noalign{\smallskip}
\multirow{3}{*}{ $\mathcal{R}_4$ } & \multirow{3}{*}{Low/High} & \multirow{3}{*}{Satisfaction}  & Feasible scale-up:  & \multirow{3}{*}{Non-negative}    \tabularnewline 
& & &  $f^{\mathsf{D}}_k(n) < f_k(n) \leq f^{\mathsf{P}}_k(n)$ at a low workload, &  \tabularnewline 
& & &  $f^{\mathsf{D}}_k(n) < f_k(n) \leq f_{\mathsf{M}}$ at a high workload &  \tabularnewline \noalign{\smallskip}\hline\noalign{\smallskip}
$\mathcal{R}_5$ & Low/High & Satisfaction & Default or scale-down: $f_k(n) \leq f^{\mathsf{D}}_k(n)$ & Positive    \tabularnewline 
\noalign{\smallskip}\hline\noalign{\vskip 0.3mm}\hline\noalign{\smallskip}
\end{tabular}
\end{table*}

Let $R_{\mathsf{M}}$, $R^{\mathsf{P}}_k(n)$, and $R^{\mathsf{D}}_k(n)$ be the minimum transmission rates for delay satisfaction if cooperative CAV pair $k$ operates at CPU frequencies $f_{\mathsf{M}}$, $f^{\mathsf{P}}_k(n)$, and $f^{\mathsf{D}}_k(n)$ respectively for processing the shared objects.  
The three rate-frequency pairs, $\left[R_{\mathsf{M}},f_{\mathsf{M}}\right]$, $\left[R^{\mathsf{P}}_k(n),f^{\mathsf{P}}_k(n)\right]$, and $\left[R^{\mathsf{D}}_k(n),f^{\mathsf{D}}_k(n)\right]$, all lie on a curve indicating $d_k(n) = \frac{\Delta}{W_k(n)}$. 
We obtain $R_{\mathsf{M}} = w \Big/\left(\frac{\Delta}{W_k(n)} - \frac{ \hat{\delta}}{f_{\mathsf{M}}}\right)$, $R^{\mathsf{P}}_k(n) = \frac{W_k(n) w \phi}{\left(\phi-1\right)\Delta}$ where $\phi = \sqrt{\frac{2\delta^3}{\hat{\delta}^2\tilde{\delta}}}>1$ is a constant, and $R^{\mathsf{D}}_k(n) = \frac{ W_k(n) w \delta     }{\left[ \left(\tilde{\rho} - \rho\right)\delta_4 - \delta_2 \right] \Delta}$, all increasing with shared workload $W_k(n)$. 
For cooperative CAV pair $k$, there are multiple 2D performance regions with different delay performance, CPU frequency scaling range, and computing efficiency gain, for different combinations of $R_k(n)$ and $f_k(n)$, as summarized in Table~\ref{Table:performance-regions}.  
The number of performance regions depends on the shared workload, as illustrated in Fig.~\ref{fig:tradeoff_together}.

From Fig.~\ref{fig:tradeoff_together} and Table~\ref{Table:performance-regions}, we obtain some useful principles in the resource allocation for each cooperative CAV pair.  
First, 
only the rate-frequency pairs in \textbf{Regions} $\mathcal{R}_4$ and $\mathcal{R}_5$ should be selected for both delay satisfaction and non-negative computing efficiency gain at a feasible CPU frequency. 
Second, a subset of \textbf{Region} $\mathcal{R}_4$ and \textbf{Region} $\mathcal{R}_5$ rate-frequency pairs that lie on curve $d_k(n) = \frac{\Delta}{W_k(n)}$ are the ideal candidate rate-frequency pairs without transmission rate over-provisioning, as indicated by red curves in Fig.~\ref{fig:tradeoff_together}. 
Accordingly, the ideal candidate rate-frequency pairs for each cooperative CAV pair $k\in\mathcal{K}_C(n)$ should satisfy the following constraints, 
\begin{align} 
    f_k(n) \leq \min \left\{ f^{\mathsf{P}}_k(n),  f_{\mathsf{M}} \right\}, \quad \forall k\in\mathcal{K}_C(n) 
    \label{eq-CPU-freq-ideal} \\
    \frac{w}{R_k(n)} + \frac{ \hat{\delta}}{f_k(n)} = \frac{\Delta}{W_k(n)}, \quad \forall k\in\mathcal{K}_C(n).
    \label{eq-delay-equality} 
\end{align}

\section{Problem Formulation}
\label{sec:problem}

Due to the dynamic radio resource availability, the current available radio resources may be insufficient for supporting all CAV pairs to work in the CP mode with delay satisfaction and non-negative computing efficiency gain.  
Thus, an adaptive set of cooperative CAV pairs, $\mathcal{K}_C(n)\subset \mathcal{K}$ should be determined for each time slot.  
Due to environmental changes, the shared workload and the transmitter-receiver distance vary over time for each cooperative CAV pair. 
As the total computing demand reduction through cooperation increases in proportion to the shared workload, a moderate shared workload increase potentially brings more computing efficiency gain. 
However, due to a reduced per-object delay budget, a heavier shared workload requires a higher CPU frequency for delay satisfaction under limited radio resources, compromising the computing efficiency gain. 
For a longer transmitter-receiver distance, the average transmission rate decreases due to a higher path loss, leading to a higher CPU frequency requirement for delay satisfaction, which reduces the computing efficiency gain.
Hence, the dynamics in both shared workloads and transmitter-receiver distances should be considered in the adaptive selection of cooperative CAV pairs, to maximize the total computing efficiency gain. 

Moreover, the cooperation decisions for each CAV pair should not change too frequently, as the switching between the SP and CP modes incur a CPU process switching overhead between scheduling the default and feature-fusion DNN models~\cite{qu2020dynamic1}.  
Let $C(n)$ be the total number of CAV pairs that change the cooperation status at time slot $n$, given by
  \begin{align} 
    C(n) = \sum_{k\in\mathcal{K}} \left| x_k(n) - x_k(n-1) \right|. 
      \label{eq-switch-cost} 
  \end{align} 
The total switching cost increases proportionally to $C(n)$.   
To maximize the total computing efficiency gain while minimizing the total switching cost in the long run, we study a joint adaptive CAV cooperation and resource allocation problem, to adaptively switch between the SP and CP modes and allocate resources among all CAV pairs for each time slot. 
Let $\boldsymbol{f}(n)=\left\{ f_k(n), \forall k\in\mathcal{K}\right\}$ be a CPU frequency allocation decision vector for processing shared objects at all CAV pairs at time slot $n$.  
Let $\dot{\boldsymbol{x}}=\left\{\boldsymbol{x}(n),\forall n\right\}$, $\dot{\boldsymbol{\beta}}=\left\{\boldsymbol{\beta}(n),\forall n\right\}$, $\dot{\boldsymbol{f}} = \left\{\boldsymbol{f}(n),\forall n\right\}$ denote the CAV cooperation, radio resource allocation, and CPU frequency allocation decisions for all time slots. 
Then, the joint problem is formulated as 
% double
\begin{align} 
    \mathbf{P}_0: ~\max_{ \dot{\boldsymbol{x}},\dot{\boldsymbol{\beta}},\dot{\boldsymbol{f}}}   \hspace{0.3cm} &{ \sum_{n} \left[\left(\sum_{k\in\mathcal{K}}G_{k}(n)\right) - \tilde{\omega} C(n)\right] }    \label{eq-obj-joint}\\
    \text{s.t. } \hspace{0.3cm}  & \eqref{eq-achieved-rate}, \eqref{eq-CPU-freq-ideal}, \eqref{eq-delay-equality} \nonumber\\
    &\sum_{k\in\mathcal{K}} \beta_k(n) \leq 1 \label{eq-trans-capacity} \\
    &0 \leq \beta_k(n) \leq x_k(n), ~ k\in\mathcal{K} \label{eq-beta-x-relationship}\\
    & \left(1-x_k(n)\right)f^{\mathsf{D}}_k(n) \leq f_k(n) \nonumber\\
    &\leq \left(1-x_k(n)\right)f^{\mathsf{D}}_k(n) + x_k(n)\mathbb{M}, ~ k\in\mathcal{K}  \label{eq-f-x-relationship}
\end{align}
where $\tilde{\omega}$ is a positive weight that controls the trade-off between gain and cost, and $\mathbb{M}$ is a very large constant.
Among the constraints, \eqref{eq-achieved-rate} is the expression of transmission rate $R_k(n)$, \eqref{eq-CPU-freq-ideal} and \eqref{eq-delay-equality} are conditions for the ideal candidate rate-frequency pairs in the 2D performance regions. 
Constraints~\eqref{eq-trans-capacity} and \eqref{eq-beta-x-relationship} ensure that the total fraction of allocated bandwidth for all CAV pairs does not exceed one, while guaranteeing that no radio resources are allocated to CAV pairs in the SP mode. 
With constraint~\eqref{eq-f-x-relationship}, we have $f_k(n) = f^{\mathsf{D}}_k(n)$ in the SP mode and $0 \leq f_k(n)\leq \mathbb{M}$ in the CP mode for CAV pair $k$.

Let $\ast$ associate the optimal solution to problem $\mathbf{P}_0$. 
Given $\dot{\boldsymbol{x}}^\ast$, the resource allocation decisions can be decoupled among time slots in the objective function and all the constraints. 
Hence, given $\dot{\boldsymbol{x}}^\ast$, $\left(\boldsymbol{\beta}^\ast(n), \boldsymbol{f}^\ast(n)\right)$ must be the resource allocation solution that maximizes the instantaneous total computing efficiency gain, $\sum_{k\in\mathcal{K}}G_{k}(n)$, for time slot $n$, since the total switching cost depends only on $\dot{\boldsymbol{x}}$. 
Therefore, problem $\mathbf{P}_0$ can be decoupled to a long-term optimization subproblem for the adaptive CAV cooperation and a series of instantaneous optimization subproblems for resource allocation, as follows.

\subsection{Resource Allocation Subproblem}
\label{sec:2D Resource Allocation Subproblem}

For time slot $n$, given any CAV cooperation decision $\boldsymbol{x}(n)$, a cooperative CAV pair set, $\mathcal{K}_C(n)$, is determined. 
For CAV pair $k$ in the SP mode, we have $\beta_k(n)=0$, $f_k(n) = f^{\mathsf{D}}_k(n)$ based on \eqref{eq-beta-x-relationship} and \eqref{eq-f-x-relationship}, and $G_{k}(n)=0$. 
Accordingly, a resource allocation subproblem for time slot $n$ is formulated as
\begin{align} 
	\mathbf{P}_1: ~ \max_{ \boldsymbol{\beta}_C(n),\boldsymbol{f}_C(n) }   \hspace{0.3cm} &{    \sum_{k\in\mathcal{K}_C(n)}G_{k}(n) }    \label{eq-total-Gain}\\
	\text{s.t. } \hspace{0.77cm}  & \eqref{eq-achieved-rate}, \eqref{eq-CPU-freq-ideal}, \eqref{eq-delay-equality} \nonumber\\
    & \sum_{k\in\mathcal{K}_C(n)} \beta_k(n) \leq 1 \label{eq-trans-capacity-P1} 
\end{align}
where $\boldsymbol{\beta}_C(n)=\{\beta_k(n),\forall k\in\mathcal{K}_C(n)\}$ and $\boldsymbol{f}_C(n)=\{f_k(n),\forall k\in\mathcal{K}_C(n)\}$ are the resource allocation decisions for cooperative CAV pairs in $\mathcal{K}_C(n)$.  
If problem $\mathbf{P}_1$ is feasible, we have $G^\ast(n)=\sum_{k\in\mathcal{K}_C(n)}G_{k}^\ast(n)$ as the maximal total computing efficiency gain achieved with optimal resource allocation; otherwise, $G^\ast(n)$ is undefined. 

\subsection{Adaptive CAV Cooperation Subproblem}
\label{sec:Adaptive CAV Cooperation Problem}

We formulate an adaptive CAV cooperation subproblem as a multi-agent Markov decision process (MMDP), where each agent corresponds to a CAV pair that makes binary cooperation decisions based on local observations over time.  
We consider cooperative agents in the MMDP, where all CAV pairs collaboratively maximize an expected total discounted reward.  
An MMDP is represented as $\left(\mathcal{K}, \mathcal{S}, \left\{\mathcal{A}_k\right\}, P, R, \left\{\Omega_k\right\}, \gamma\right)$, 
where $\mathcal{K}$ is a set of agents, 
$\mathcal{S}$ is the state space, $\mathcal{A}_k$ is the action space for agent $k$, with $\mathcal{A} = \times_k \mathcal{A}_k$ being the set of joint actions, $P$ is an unknown state transition probability matrix, 
$R: \mathcal{S}\times\mathcal{A}\mapsto \mathbb{R}$ is a reward function, 
$\Omega_k$ is a set of observations for agent $k$,
and $\gamma$ is a discount factor in $[0,1)$.   
The observation, action, and reward of agent $k$ are given as follows.

\begin{itemize}

	\item \emph{Observation}: 
    The local observation for agent $k$ at time slot $n$, denoted by $o_k^{(n)}$, includes the available radio spectrum bandwidth for CAVs, $B(n)$, shared workload $W_k(n)$, transmitter-receiver distance $D_k(n)$, and previous cooperation status $x_k(n-1)$, given by
	\begin{align} 
		o_k^{(n)} = \left\{ B(n), W_k(n), D_k(n), x_k(n-1) \right\};
	    \label{eq-original-obs} 
	\end{align}

	\item \emph{Action}: At time slot $n$, agent $k$'s action is the cooperation decision, $x_k(n)$.   
	Joint action $\boldsymbol{x}(n)=\{x_k(n),\forall k\}$ determines a cooperative CAV pair set, $\mathcal{K}_C(n)$; 

	\item \emph{Reward}: 
	The reward for any agent at time slot $n$, denoted by $r^{(n)}$, is given by
    % double
	\be 
	r^{(n)} =  \left\{ \hspace{0mm}
	\begin{array}{rcl}
	G^\ast(n) - \tilde{\omega} C(n)   \text{,}      & {\mbox{if $\mathbf{P}_1$ is feasible}}\\
	{P \text{,} \hspace{20mm}}    & {\mbox{otherwise\hspace{10mm}}}
	\label{eq-sys-reward}
	\end{array} \right.
	\ee
	where $G^\ast(n)$ is the maximal total computing efficiency gain associated with optimal resource allocation for $\mathcal{K}_C(n)$. 
    If problem $\mathbf{P}_1$ is infeasible under the selected joint action, a negative penalty, $P$, is used as the reward. 
    As both the objective function and the delay constraint in problem $\mathbf{P}_1$ are derived based on the early exit probabilities of DNN models, we aim to learn the \emph{statistically} optimal adaptive CAV cooperation actions in the long run by using the reward function in \eqref{eq-sys-reward}.

\end{itemize}

\section{Model-Assisted Multi-Agent Reinforcement Learning Solution}
\label{sec:solution}

We propose a model-assisted learning solution to the joint problem. First, a model-based optimal solution is derived for the resource allocation subproblem. Then, a model-assisted MARL approach that relies on the model-based optimal resource allocation solution for reward calculation is proposed, to solve the MMDP for adaptive CAV cooperation. 

\subsection{Optimal Resource Allocation Solution}
\label{sec:KKT Solution for 2D Resource Allocation Optimization}

With~\eqref{eq-reduced-energy}, maximizing $\sum_{k\in\mathcal{K}_C(n)}G_{k}(n)$ in $\mathbf{P}_1$ is equivalent to minimizing $\sum_{k\in\mathcal{K}_C(n)}  W_k(n)f_k(n)^2$.
By writing $\beta_k(n)$ as a function of $R_k(n)$, we combine constraints \eqref{eq-achieved-rate} and \eqref{eq-trans-capacity-P1} as %one constraint, given by
\begin{align} 
	\sum_{k\in\mathcal{K}_C(n)} \frac{R_k(n)}{B(n) \log_2 \left( 1+ p_k g_k(n) / \sigma^2 \right)} \leq 1.
    \label{eq-rate-capacity} 
\end{align}
From constraint \eqref{eq-delay-equality}, $R_k(n) = w \big/\left(\frac{\Delta}{W_k(n)} - \frac{ \hat{\delta}}{f_k(n)}\right)$ is a function of $f_k(n)$. 
Substituting $R_k(n)$ in \eqref{eq-rate-capacity}, we transform \eqref{eq-rate-capacity} into a constraint on decision variables $\boldsymbol{f}_C(n)$, given by
\begin{align} 
	h( \boldsymbol{f} ) = \sum_{k\in\mathcal{K}_C(n)} \frac{c_k}{b_k - \hat{\delta}/f_k(n) }  - 1 \leq 0
    \label{eq-CPU-total-capacity} 
\end{align}
where $b_k = \frac{\Delta}{W_k(n)}$ and  $c_k = \frac{w}{B(n) \log_2 \left( 1+ p_k g_k(n) / \sigma^2 \right)}$ are known parameters given network status for time slot $n$.  
Here, $h( \boldsymbol{f} )$ is a monotonically decreasing constraint function of $\boldsymbol{f}_C(n)$, defined in domain $\{f_k(n) >  \frac{\hat{\delta}}{b_k}, \forall k\in\mathcal{K}_C(n)\}$ to ensure $R_k(n)>0$ for $k\in\mathcal{K}_C(n)$. 
Let $f^0_k(n) = \min \left\{ f^{\mathsf{P}}_k(n),  f_{\mathsf{M}} \right\}$, which is known to CAV pair $k$, 
and let $\boldsymbol{f}_C^0(n)=\left\{ f^0_k(n), \forall k\in\mathcal{K}_C(n)\right\}$. 
Then, problem $\mathbf{P}_1$ is transformed to a CPU frequency allocation problem, given by
\begin{align} 
	\mathbf{P}_2: \hspace{0.4cm} \min_{ \boldsymbol{f}_C(n) }   \hspace{0.4cm} &{   \sum_{k\in\mathcal{K}_C(n)} W_k(n) f_k(n)^2  }   
	\label{eq-obj} 
	\\
	\text{s.t. } \hspace{0.39cm}  
 	& \boldsymbol{f}_C(n) \preceq \boldsymbol{f}_C^0(n) %f_k(n) \leq f^0_k(n), \quad \forall k\in\mathcal{K}_C(n) 
 	\label{eq-CPU-upper-bound}
 	\\
    & h( \boldsymbol{f} ) \leq 0.
    \label{eq-constraint-func-inequality} 
\end{align}
\begin{Thm}
\emph{Problem $\mathbf{P}_2$ is convex, and strong duality holds if the problem is feasible under condition $h( \boldsymbol{f}^0 ) < 0$.} 
\end{Thm}
The proof of Theorem 1 is given in Appendix. 
Based on Theorem 1, Karush-Kuhn-Tucker (KKT) conditions are necessary and sufficient conditions for optimal solution to problem $\mathbf{P}_2$~\cite{boyd2004convex}.
The Lagrangian of $\mathbf{P}_2$ is 
% double
\begin{align} 
	\mathcal{L}\left( \boldsymbol{f}, \boldsymbol{\lambda}, \nu \right) = &\sum_{k\in\mathcal{K}_C(n)} W_k(n)f_k(n)^2  \nonumber\\
	& + \sum_{k\in\mathcal{K}_C(n)} \lambda_k \left( f_k(n) - f^0_k(n)  \right) + \nu \hspace{0.6mm} h( \boldsymbol{f} )
    \label{eq-Lagrangian}
\end{align}
where $\boldsymbol{\lambda}=\{\lambda_k, \forall k\in\mathcal{K}_C(n)\}$ and $\nu$ are dual variables. 
Its gradient with respect to primal variable $f_k(n)$ is given by
% double
\begin{align} 
	\frac{\partial \mathcal{L}\left( \boldsymbol{f}, \boldsymbol{\lambda}, \nu \right)}{\partial f_k(n)} \hspace{-0.5mm}= \hspace{-0.5mm} 2 W_k(n)f_k(n) \hspace{-0.5mm}+\hspace{-0.5mm} \lambda_k \hspace{-0.5mm}-\hspace{-0.5mm} \nu \frac{c_k \hat{\delta}}{\left( b_k f_k(n) - \hat{\delta}\right)^2}. \hspace{-2mm}
    \label{eq-Lag-gradient}
\end{align}
Let $\boldsymbol{f}^\ast$ be a primal optimal point and $\left(\boldsymbol{\lambda}^\ast, \nu^\ast\right)$ be a dual optimal point. 
Then, the KKT conditions for the optimal solution to problem $\mathbf{P}_2$ are given by
\begin{subequations}  %double
\be
    & & \hspace{-1.3cm}
    f_k^\ast \leq f^0_k, \quad \forall k\in\mathcal{K}_C
    \label{KKT-1}
	\\
    & & \hspace{-1.3cm}
 	h( \boldsymbol{f}^\ast ) \leq 0 
    \label{KKT-2}
    \\
	& & \hspace{-1.3cm}
    \lambda_k^\ast \geq 0, \quad \forall k\in\mathcal{K}_C
    \label{KKT-3}
    \\
    & & \hspace{-1.3cm}
 	\nu^\ast \geq 0
    \label{KKT-4}
    \\
    & & \hspace{-1.3cm}
 	\lambda_k^\ast \left( f_k^\ast - f^0_k\right) = 0, \quad \forall k\in\mathcal{K}_C
    \label{KKT-5}
    \\
    & & \hspace{-1.3cm}
 	2 W_k f_k^\ast + \lambda_k^\ast - \nu^\ast \frac{c_k \hat{\delta}}{\left( b_k f_k^\ast - \hat{\delta}\right)^2} = 0, ~ \forall k\in\mathcal{K}_C
    \label{KKT-6}
\ee
\label{P-resource-subroutine}
\end{subequations}
% }
\hspace{-3.5mm} where time slot index $n$ is omitted for brevity. 
Among all conditions, \eqref{KKT-1} and \eqref{KKT-2} are primal feasible conditions, \eqref{KKT-3} and \eqref{KKT-4} are dual feasible conditions, \eqref{KKT-5} indicates the complementary slackness, and \eqref{KKT-6} ensures that the Lagrangian gradient in \eqref{eq-Lag-gradient} vanishes at $\boldsymbol{f}^\ast$ as $\boldsymbol{f}^\ast$ minimizes $\mathcal{L}\left( \boldsymbol{f}, \boldsymbol{\lambda}^\ast, \nu^\ast \right)$~\cite{boyd2004convex}.
From \eqref{KKT-5} and \eqref{KKT-6}, we obtain 
% double
\begin{align} 
	 \left(  \frac{\nu^\ast c_k \hat{\delta}}{\left( b_k f_k^\ast - \hat{\delta}\right)^2} - 2 W_k f_k^\ast \right) \left( f_k^\ast - f^0_k\right) = 0,  \forall k\in\mathcal{K}_C.   
    \label{eq-KKT-5-6}
\end{align}

% double
\begin{figure}
\centering
{
\includegraphics[width=0.8\linewidth]{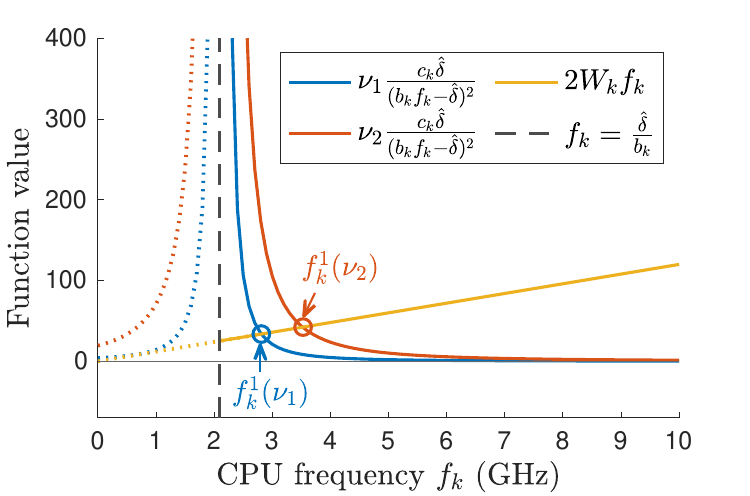}} 
\caption{An illustration of the root of function $S\left(f_k, \nu\right)$ for $\nu_1<\nu_2$.}\label{fig:S_func_root}
\end{figure}

For dual variable $\nu$, let $S\left(f_k, \nu\right) = \nu \frac{  c_k \hat{\delta}}{\left( b_k f_k - \hat{\delta}\right)^2} - 2 W_k f_k$, which is a monotonically decreasing function in domain $f_k >  \frac{\hat{\delta}}{b_k}$. 
Let $f_k^1(\nu)$ be the root of $S\left(f_k, \nu\right)$, 
which corresponds to the intersection point of functions $\nu \frac{c_k \hat{\delta}}{\left( b_k f_k - \hat{\delta}\right)^2}$ and $2 W_k f_k$ in domain $f_k >  \frac{\hat{\delta}}{b_k}$. 
Fig.~\ref{fig:S_func_root} illustrates the root of function $S\left(f_k, \nu\right)$ for $\nu_1<\nu_2$. 
We see that, for a smaller value of dual variable $\nu$, the root of $S\left(f_k, \nu\right)$, i.e., $f_k^1(\nu)$, has a smaller value. 
According to \eqref{KKT-6}, we have $\lambda_k^\ast = S\left(f_k^\ast, \nu^\ast\right)$, and condition \eqref{eq-KKT-5-6} is rewritten as
\begin{align} 
	 S\left(f_k^\ast, \nu^\ast\right) \left( f_k^\ast - f^0_k\right) = 0, \quad \forall k\in\mathcal{K}_C.   
    \label{eq-KKT-5-6-rewritten}
\end{align} 
Then, the primal optimal point, $\boldsymbol{f}^\ast$, that minimizes the objective function in \eqref{eq-obj} is given by
\be % double
f_k^\ast =  \left\{ \hspace{0mm}
\begin{array}{rcl}
f_k^1(\nu^\ast)  \text{,}      & {\mbox{if } f_k^1(\nu^\ast) \leq f_k^0 }\\
f_k^0 \text{,\hspace{0.65cm}}   & {\mbox{if } f_k^0 < f_k^1(\nu^\ast) }
\label{eq-f_opt}
\end{array} \right.
\ee
for any $k\in\mathcal{K}_C$, to guarantee condition \eqref{eq-KKT-5-6-rewritten} and ensure that $\lambda_k^\ast = S\left(f_k^\ast, \nu^\ast\right) \geq 0$ in \eqref{KKT-3}. 
Specifically, if $f_k^1(\nu^\ast) \leq f_k^0$, we have $\lambda_k^\ast=0$; otherwise, we have $\lambda_k^\ast>0$.

As $\nu^\ast$ is an unknown optimal dual variable in~\eqref{eq-f_opt}, we will find the primal optimal point, $\boldsymbol{f}^\ast$, by exploring the possible values of $\nu^\ast$. 
Let $\hat{\nu}^\ast$ be a candidate value of $\nu^\ast$, and let $\hat{\boldsymbol{f}}^\ast$ be the corresponding candidate value of $\boldsymbol{f}^\ast$ for $\nu^\ast = \hat{\nu}^\ast$ based on~\eqref{eq-f_opt}.  
For a smaller $\hat{\nu}^\ast$ value, $\hat{f}_k^1(\hat{\nu}^\ast)$ is smaller, then $\hat{f}_k^\ast$ is potentially smaller according to \eqref{eq-f_opt}, leading to a potentially smaller objective value in \eqref{eq-obj}. 
However, as $h( \hat{\boldsymbol{f}}^\ast )$ is a decreasing function of $\hat{f}_k^\ast$, the primal feasible condition in \eqref{KKT-2} might be violated if the value of $\hat{\nu}^\ast$ is too small. 
Therefore, 
we should find a minimum non-negative value for $\hat{\nu}^\ast$ that satisfies \eqref{KKT-2} and \eqref{KKT-4}, to obtain $\nu^\ast$ and $\boldsymbol{f}^\ast$.

To solve problem $\mathbf{P}_2$, the feasibility is first checked by calculating $h( \boldsymbol{f}^0 )$. 
If $h( \boldsymbol{f}^0 ) > 0$, the problem is infeasible; if $h( \boldsymbol{f}^0 ) = 0$, we can directly obtain the optimal solution as $\boldsymbol{f}^\ast = \boldsymbol{f}^0$; if $h( \boldsymbol{f}^0 ) < 0$, we continue to use a binary-search method to iteratively find $\nu^\ast$ in a gradually reduced interval $[\nu_L, \nu_R]$. 
The detailed algorithm is described as follows. 

	For initialization, as $f_k^\ast \in \left(\frac{\hat{\delta}}{b_k}, f_k^0\right]$, $\nu_L$ and $\nu_R$ are set to satisfy $S\left(\frac{\hat{\delta}}{ \max_{k\in\mathcal{K}_C} b_k }+\epsilon, \nu_L \right) = 0$ and  $S\left( \min_{k\in\mathcal{K}_C} f_k^0, \nu_R\right) = 0$. 
	Here, $\epsilon$ is a very small number satisfying $0< \epsilon \ll 1$. 
	The shared workloads among all cooperative CAV pairs determine the initial binary-search interval, $[\nu_L, \nu_R]$.   
	In each iteration, $\hat{\nu}^\ast$ is set as $\frac{\nu_L+\nu_R}{2}$, and $\hat{\boldsymbol{f}}^\ast$ is obtained based on \eqref{eq-f_opt}. 
	For $\boldsymbol{f} = \hat{\boldsymbol{f}}^\ast$, let $\hat{G}^\ast = \sum_{k\in\mathcal{K}_C}\hat{G}^\ast_{k}$ denote the corresponding total computing efficiency gain in \eqref{eq-total-Gain}. 
	Constraint \eqref{KKT-2} is checked by calculating $h( \hat{\boldsymbol{f}}^\ast )$. Whether or not the binary search ends at the current iteration and how to update interval $[\nu_L, \nu_R]$ for the next iteration depend on $h( \hat{\boldsymbol{f}}^\ast )$: %, which are as follows. 

	\begin{itemize}

	\item If $h( \hat{\boldsymbol{f}}^\ast )=0$, the optimal points, $\nu^\ast$ and $\boldsymbol{f}^\ast$, are obtained as $\hat{\nu}^\ast$ and $\hat{\boldsymbol{f}}^\ast$, and the algorithm is ideally finished. In practice, 
	we set a stopping criteria, $-10^{-4}<h( \hat{\boldsymbol{f}}^\ast )<0$, and obtain asymptotically optimal primal and dual variables when the binary search ends; 

	\item If $h( \hat{\boldsymbol{f}}^\ast )>0$, constraint \eqref{KKT-2} is infeasible, and the candidate value for $\nu^\ast$ should be increased in the next iteration, thus we set $\nu_L = \hat{\nu}^\ast$; 

	\item If $h( \hat{\boldsymbol{f}}^\ast )<0$, the candidate value for $\nu^\ast$ can be further reduced in the next iteration to increase $h( \hat{\boldsymbol{f}}^\ast )$ and reduce the objective value in \eqref{eq-obj}, thus we set $\nu_R = \hat{\nu}^\ast$. 

	\end{itemize}

Such a centralized iterative algorithm can be executed at the cluster head which is responsible for collecting the overall network dynamics in the vehicle cluster.

\subsection{Model-Assisted Multi-Agent Reinforcement Learning}

We use a model-assisted MARL algorithm to solve the MMDP for adaptive CAV cooperation. 
To address the inherent nonstationary issue due to the partial observability at each agent, we use a multi-agent deep deterministic policy gradient (MADDPG) algorithm which adopts a centralized training distributed execution (CTDE) framework~\cite{lowe2017multi,tian2021multiagent}.
The model-assisted MADDPG algorithm is presented in Algorithm~\ref{alg:MADDPG}.  
Each agent trains a critic network and an actor network based on the global state and joint action in a centralized training stage, and uses the trained actor network for decision based on local observation in a distributed execution stage. 
To enhance the observability at each agent, we augment the local observation by two additional elements, i.e., the average workload and the average transmitter-receiver distance among all agents, both of which can be provided by the cluster head~\cite{ye2021joint}.
In this manner, a CAV pair has both local information and some statistical global information for better-informed decisions during the distributed execution stage, without acquiring the full global state.       
Let $s_k^{(n)}$ be the augmented local observation for agent $k$ at time slot $n$, given by
$s_k^{(n)}  = \left\{ o_k^{(n)}, \frac{\sum_{k\in\mathcal{K}} W_{k}(n) }{K}, \frac{\sum_{k\in\mathcal{K}} D_{k}(n) }{K}\right\}$. %, 
Let $\boldsymbol{s}^{(n)} = \{ s_k^{(n)}, \forall k\in\mathcal{K} \}$ be the global state at time slot $n$. 
   
As the MMDP has a discrete action space for each agent, specifically a binary action space for whether or not to cooperate, a Gumbel-Softmax estimator is used by each agent to allow the backpropagation of gradients through the actor network during training~\cite{lowe2017multi}. 
For agent $k$, an actor network, $\mu_k(s_k)$, parameterized by weights $\boldsymbol{\varphi}_k$, is trained to learn a continuous action, $a_k = \{a_{k,j}, j = 0,1\}$, based on augmented local observation $s_k$. Here, $a_k$ is an estimated two-dimension Gumbel-Softmax distribution over the binary action space~\cite{jang2016categorical}. 
Let $a_k^{(n)} =\{a_{k,j}^{(n)}, j = 0,1\}$ be the continuous action of agent $k$ at time slot $n$, and let $\boldsymbol{a}^{(n)}=\{a_k^{(n)},\forall k\}$ be the joint continuous action at time slot $n$. 
Agent $k$ obtains the cooperation decision as the binary action with the maximum Gumbel-Softmax probability, i.e., $x_k(n) = \arg \max_{j=0,1} a_{k,j}^{(n)}$.

In addition to actor network $\mu_k(s_k)$, agent $k$ trains a critic network parameterized by weights $\boldsymbol{\theta}_k$ to approximate a centralized $Q$-function, $Q_k(\boldsymbol{s},\boldsymbol{a}) = \mathbb{E}\left[\sum_{n=0}^{N-1} \gamma^{n} r_k^{(n)} \big| \boldsymbol{s}, \boldsymbol{a}   \right]$, that takes global state $\boldsymbol{s}$ and joint action $\boldsymbol{a}$ as input to estimate a $Q$-value, where $N$ is the maximum number of learning steps in an episode. 
In the training stage of such an actor-critic framework, although the agents independently take actions based on the augmented local observations, they evaluate the actions and refine the policies by taking the actions of other agents into consideration in $Q_k(\boldsymbol{s},\boldsymbol{a})$, thus facilitating a collaborative exploration of the vehicular network environment to maximize a collective reward. 
To overcome the divergence update issue, agent $k$ also has a target critic network, $\hat{Q}_k(\boldsymbol{s},\boldsymbol{a})$, parameterized by $\hat{\boldsymbol{\theta}}_k$, and a target actor network, $\hat{\mu}_k(s_k)$, parameterized by $\hat{\boldsymbol{\varphi}}_k$, with delayed updates. % in comparison with the primary critic and actor networks. 
Agent $k$ initializes the primary and target critic and actor networks with random weights before training (\emph{line} 1) and then continually updates the weights until convergence.
MADDPG employs a soft updating strategy, where agent $k$ updates weights $\hat{\boldsymbol{\theta}}_k$ and $\hat{\boldsymbol{\varphi}}_k$ of the target networks in each learning step (\emph{line} 13) as
\begin{align} 
	 \hat{\boldsymbol{\theta}}_k = \xi \boldsymbol{\theta}_k + \left(1-\xi \right)\hat{\boldsymbol{\theta}}_k ~~\mbox{and}~~
	 \hat{\boldsymbol{\varphi}}_k = \xi \boldsymbol{\varphi}_k + \left(1-\xi \right)\hat{\boldsymbol{\varphi}}_k 
    \label{eq-target-update}
\end{align}
with $\xi$ being the soft updating rate of the target networks. 

\begin{algorithm}[t]
\small
\DontPrintSemicolon
\SetKwInOut{Input}{Input}{}
\SetKwInOut{Output}{Output}
\SetKwInput{Define}{Define}
\SetKwInput{Let}{Let}
\SetKwInput{Find}{Find}
\tcc{Centralized Training Stage}
All agents initialize networks with random weights. \\
\For{each episode }{ 
    Initialize local observation $o_k^{(0)}$ for $k\in\mathcal{K}$.\\ 
    \For{learning step $n$}{  
        \For{agent $k$}{ 
            Send local observation $o_k^{(n)}$ to cluster head.\\
            Collect augmented local observation $s_k^{(n)}$. \\
            Decide continuous action $a_k^{(n)} = \mu_k\left(s_k^{(n)}\right)$, derive binary cooperation decision $x_k(n)$, and send $a_k^{(n)}$ to cluster head.\\
        }
        Cluster head solves the resource allocation subproblem, obtains $G^\ast(n)$ and $r^{(n)}$, and broadcasts $\boldsymbol{s}^{(n)}$, $\boldsymbol{a}^{(n)}$, and $r^{(n)}$ to all agents.  \\ 
        \For{agent $k$}{ 
            Add $\left(\boldsymbol{s}^{(n-1)}, \boldsymbol{a}^{(n-1)}, r^{(n-1)}, \boldsymbol{s}^{(n)}\right)$ to buffer $\mathcal{B}$ if $n\geq 1$.\\ 
            Sample mini-batch of experiences from $\mathcal{B}$. \\
            Update primary and target critic and actor networks based on \eqref{eq-critic-update}, \eqref{eq-actor-update}, and \eqref{eq-target-update}.
        }
    }
}   
\tcc{Distributed Execution Stage}
\For{each episode }{ 
    \For{learning step $n$}{  
        \For{agent $k$}{ 
            Send local observation $o_k^{(n)}$ to cluster head.\\
            Collect augmented local observation $s_k^{(n)}$. \\
            Decide continuous action $a_k^{(n)} = \mu_k\left(s_k^{(n)}\right)$, derive binary cooperation decision $x_k(n)$, and send $x_k(n)$ to cluster head.\\
        }
        Cluster head allocates resources to cooperative CAV pairs.  \\ 
    }
}   
\caption{A Model-Assisted MADDPG Algorithm}
\label{alg:MADDPG}
\end{algorithm}

The agents interact with the network environment in a sequence of episodes, each containing a finite number of learning steps, one learning step for one time slot. 
An episode starts when a vehicle cluster is about to move into an RSU's coverage area and ends once it leaves the coverage. 
At the beginning of each episode, each agent initializes the local observation (\emph{line} 3).
At the beginning of time slot $n$, each agent $k$ sends local observation $o_k^{(n)}$ via a dedicated control channel to the cluster head (\emph{line} 6), which calculates and then returns the augmented information.  
Agent $k$ collects augmented local observation $s_k^{(n)}$ (\emph{line} 7), based on which the agent decides continuous action $a_k^{(n)}$ as $\mu_k\left(s_k^{(n)}\right)$ by the primary actor network, and then discretizes it into a binary cooperation decision, $x_k(n)$ (\emph{line} 8). 
All agents send the continuous actions to the cluster head (\emph{line} 8).  
The cluster head allocates resources to cooperative CAV pairs, determines the maximal total computing efficiency gain, $G^\ast(n)$, using the model-based resource allocation solution, and calculates reward $r^{(n)}$ in \eqref{eq-sys-reward} which is then broadcast to all agents together with global state $\boldsymbol{s}^{(n)}$ and joint action $\boldsymbol{a}^{(n)}$ (\emph{line} 9). 
Then, each agent adds a new transition tuple $\left(\boldsymbol{s}^{(n-1)}, \boldsymbol{a}^{(n-1)}, r^{(n-1)}, \boldsymbol{s}^{(n)}\right)$ to an experience replay buffer, $\mathcal{B}$, if $n\geq 1$ (\emph{line} 11).

To train the critic and actor networks at each learning step, each agent samples a mini-batch of $I$ experiences from $\mathcal{B}$, among which $\left(\boldsymbol{s}^{(i)}, \boldsymbol{a}^{(i)}, r^{(i)}, \boldsymbol{s}^{(i+1)}\right)$ represents the $i$-th experience (\emph{line} 12). 
Agent $k$ updates the critic network by minimizing a loss function, $\mathbb{L}_k(\boldsymbol{\theta}_k) = \frac{1}{I}\sum_{i=1}^I \left[ y_k^{(i)} - Q_k\left( \boldsymbol{s}^{(i)},\boldsymbol{a}^{(i)} \right) \right]^2$,  
where $y_k^{(i)} = r^{(i)} + \gamma \hat{Q}_k\left( \boldsymbol{s}^{(i+1)}, \hat{\mu}_1\left( s_1^{(i+1)} \right), \dots, \hat{\mu}_K\left( s_K^{(i+1)} \right) \right)$ is a target value estimated by the target critic and actor networks. 
Weights $\boldsymbol{\theta}_k$ are updated via a gradient descent (\emph{line} 13), given by
\begin{align} 
	\boldsymbol{\theta}_k \leftarrow \boldsymbol{\theta}_k - \alpha_{\boldsymbol{\theta}}  \nabla_{\boldsymbol{\theta}_k} \mathbb{L}_k(\boldsymbol{\theta}_k)
    \label{eq-critic-update}
\end{align}
where $\alpha_{\boldsymbol{\theta}}$ is the learning rate for critic networks. 
The actor network of agent $k$ aims to maximize a long-term total expected reward, $J_k(\boldsymbol{\varphi}_k) = \mathbb{E}\left[\sum_{n=0}^{N} \gamma^{n} r_k^{(n)}\right]$, which is the expected $Q$-value among all state-action pairs, i.e.,  $J_k(\boldsymbol{\varphi}_k) = \mathbb{E}_{\boldsymbol{s},\boldsymbol{a}} Q_k(\boldsymbol{s},\boldsymbol{a})$. 
Thus, 
weights $\boldsymbol{\varphi}_k$ of the actor network are updated via a gradient ascent (\emph{line} 13), given by
\begin{align} 
	\boldsymbol{\varphi}_k \leftarrow \boldsymbol{\varphi}_k + \alpha_{\boldsymbol{\varphi}}  \nabla_{\boldsymbol{\varphi}_k} J_k(\boldsymbol{\varphi}_k)
    \label{eq-actor-update}
\end{align}
where $\alpha_{\boldsymbol{\varphi}}$ is the learning rate for actor networks, and the gradient of $J_k(\boldsymbol{\varphi}_k)$ is given by
% double
\begin{align} 
	\nabla_{\boldsymbol{\varphi}_k} J_k(\boldsymbol{\varphi}_k) = &\frac{1}{I}\sum_{i=1}^I \left[\nabla_{\boldsymbol{\varphi}_k} \mu_k\left( s_k^{(i)} \right) \nabla_{a_k} Q_k\left( \boldsymbol{s}^{(i)}, a_1^{(i)}, \right.\right. \nonumber\\ % \cdot
	& \left.\left. \dots, a_k, \dots, a_K^{(i)} \right) \Big|_{a_k = \mu_k\left( s_k^{(i)} \right)} \right]. 
\end{align} 

After the centralized training stage, each agent uses the trained actor network for decision in a distributed execution stage, with the assistance of the cluster head (\emph{lines} 14-20).

\section{Simulation Results}
\label{sec:Simulation Results}

\subsection{Simulation Setup}

\begin{table}[t]
\small
\centering
\caption{\scshape{System parameters in simulation}}
\begin{tabular}{ p{0.62\columnwidth} | c } 
\hline\noalign{\vskip 0.3mm}\hline\noalign{\smallskip}
Parameters & Value \\
\noalign{\smallskip}\hline\noalign{\smallskip}
Center frequency ($f_c$)&  $6$ GHz \tabularnewline
Noise power ($\sigma^2$) & $-104$ dBm \tabularnewline 
Transmit power ($p_k$) & $23$ dBm \tabularnewline 
Maximum local CPU frequency ($f_{\mathsf{M}}$) & $8$ GHz \tabularnewline
Energy efficiency coefficient ($\kappa$) & $10^{-28}$ $\mathsf{J}/\mathsf{s}/\mathsf{Hz}^3$ \tabularnewline
Feature extraction computing demand ($\delta_1$) & $4\times 10^{6}$ cycles \tabularnewline
Feature fusion computing demand ($\delta_2$) & $1000$ cycles \tabularnewline
Fast inference computing demand ($\delta_3$) & $3.1\times 10^{5}$ cycles \tabularnewline
Full inference computing demand ($\delta_4$) & $7.7\times 10^{7}$ cycles \tabularnewline
Feature data size ($w$) & $0.29$ Mbits \tabularnewline
Default early exit probability ($\rho$) & $0.3$ \tabularnewline
Feature-fusion early exit probability ($\tilde{\rho}$) & $0.6$ \tabularnewline
\hline\noalign{\vskip 0.3mm}\hline\noalign{\smallskip}
\end{tabular}
\label{Table:System parameters in simulation}
\end{table}

We consider a four-lane unidirectional highway, where $K\in\{2,3,4,5,6\}$ CAV pairs are moving together with $10$ HDVs in a vehicle cluster, with an intermittent RSU coverage.  
The RSU is $10$ $m$ away from the highway and provides a communication radius of 250 $m$.
In the RL task, we consider a $1500$ $m$ highway segment for each episode, during which a vehicle cluster moves for a distance of $1000$ $m$ through the RSU coverage. 
The vehicle speed is uniformly set in a range of $[23,27]$ $m/s$. 
The time slot length is $500$ $ms$. Each episode spans over an average time duration of $40$ $s$ and contains $80$ time slots on average. 
The delay requirement for object classification in each time slot is $\Delta = 100$ $ms$. 

We use the Simulation of Urban Mobility (SUMO) traffic simulator to simulate the vehicle trajectories in each episode, based on which the transmitter-receiver distances of each CAV pair can be obtained~\cite{SUMO,ye2021joint,abdel2021vehicular}. 
To obtain the time-varying channel power gain for each CAV pair, we use the 3GPP NR-V2X 37.885 highway case for the V2V link path loss calculation~\cite{3gpp.37.885}. %,zhang2019deep
The path loss in $dB$ for CAV pair $k$ during time slot $n$ is calculated as $L_{dB}(n) = 32.4 + 20\log_{10} D_k(n) + 20\log_{10} f_c$, 
where $D_k(n)$ is the transmitter-receiver distance in meter, and $f_c$ is the center frequency in GHz.
For CAV pair $k$, the transitions of shared workload, $W_k(n)$, across different time slots follow a Markov chain with states in $\{4,5,6,7,8\}$. 
When residing in the RSU coverage,  
each HDV generates V2R transmission requests in each time slot according to a $Bernoulli(0.5)$ distribution, with each V2R transmission request occupying a bandwidth of $0.5$ MHz.  
With a total bandwidth of $B=10.5$ MHz for V2X sidelink communication, the available radio spectrum bandwidth for V2V transmission at time slot $n$ is $B(n) = B - 0.5 M(n)$, where $M(n)$ is the number of HDVs that request V2R transmission at time slot $n$. 
On average, the $B(n)$ value in an episode follows a decreasing-then-increasing trend when the vehicle cluster drives through the RSU coverage. 
Other system parameters are given in Table~\ref{Table:System parameters in simulation}.

We implement both the iterative algorithm for optimal resource allocation and the MADDPG algorithm for adaptive CAV cooperation using Python 3.9.2. The learning modules are implemented using TensorFlow 2.11.0.   
Each learning agent has two hidden layers with $(64,64)$ neurons and \texttt{Relu} activation functions in critic and actor networks. The critic network has a one-dimension output with no activation function, and the actor network has a two-dimension output with $\texttt{Gumbel-SoftMax}$ activation. 
We set weight $\tilde{\omega} \in [0,1]$  
and penalty $P = -10$ in reward function \eqref{eq-sys-reward}, and use $\alpha_{\boldsymbol{\theta}} = 10^{-2}$ and $\alpha_{\boldsymbol{\varphi}}=10^{-3}$ as the critic and actor learning rates, a soft updating rate of $\xi=0.01$ for target network update, and discount factor $\gamma=0.95$. 
For training at each learning step, a mini-batch of $I=1024$ experiences are sampled from buffer $\mathcal{B}$ with size $100000$.

\subsection{Performance Evaluation}

We first evaluate the performance of the optimal resource allocation solution, for a set, $\mathcal{K}_C$, of cooperative CAV pairs with an available bandwidth of $10.5$ MHz for V2V transmission. 
Without loss of generality, all cooperative CAV pairs have an identical shared workload, $W$. 
We examine the impact of $|\mathcal{K}_C|$ and $W$ on the total computing efficiency gain. 

\begin{figure}
    \centering
    \begin{subfigure}[b]{0.38\textwidth}
      \includegraphics[width=1\linewidth]{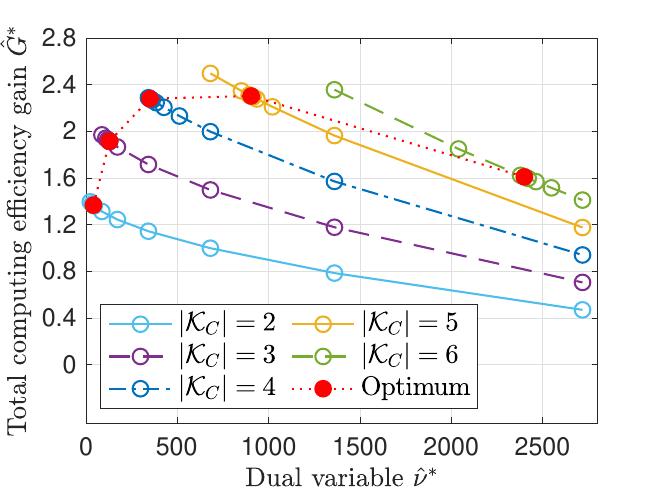} 
      \caption{}\label{fig:KKT_vs_mu_homo_CAVnum_obj}
    \end{subfigure}
    ~~~ 
    \begin{subfigure}[b]{0.38\textwidth}
  \includegraphics[width=1\linewidth]{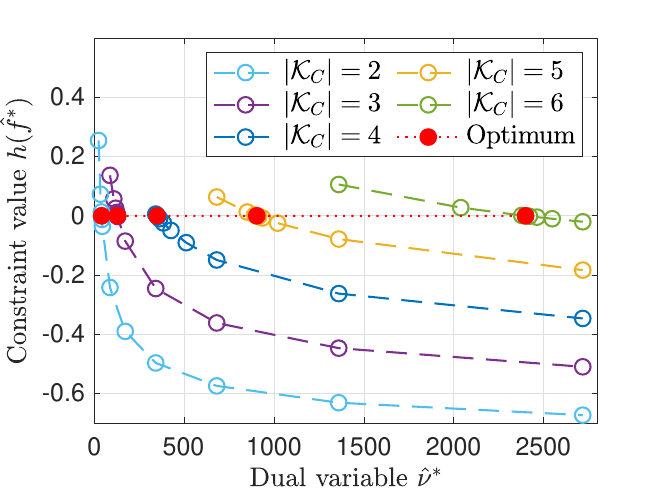} 
  \caption{}\label{fig:KKT_vs_mu_homo_CAVnum_const}
    \end{subfigure}
  \caption{Performance of the optimal resource allocation solution for a different number of cooperative CAV pairs ($|\mathcal{K}_C|$) at $W=6$. (a) Total computing efficiency gain. (b) Constraint value. }\label{fig:KKT_CAVnum}
\end{figure}

\begin{figure}
    \centering
    \begin{subfigure}[b]{0.38\textwidth}
      \includegraphics[width=1\linewidth]{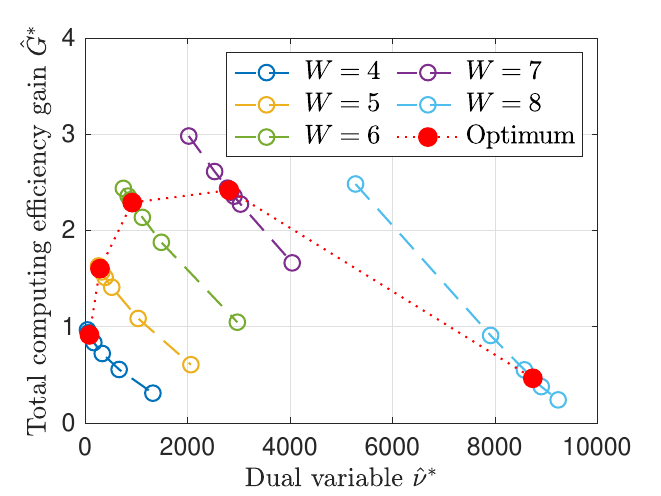} 
      \caption{}\label{fig:KKT_vs_mu_homo_workload_obj}
    \end{subfigure}
    ~~~ 
    \begin{subfigure}[b]{0.38\textwidth}
  \includegraphics[width=1\linewidth]{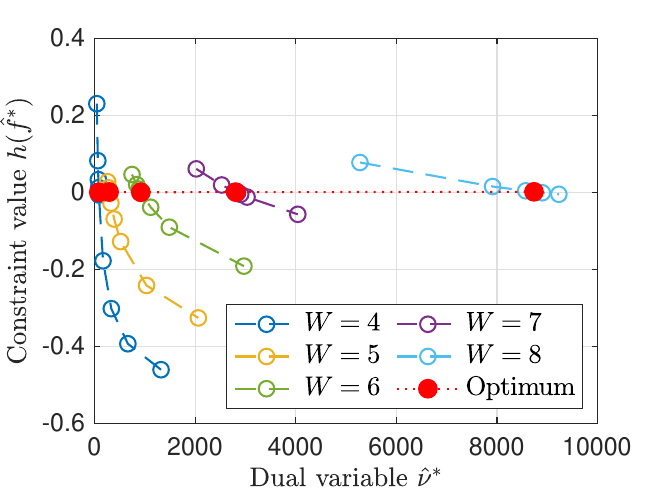} 
  \caption{}\label{fig:KKT_vs_mu_homo_workload_const}
    \end{subfigure}
  \caption{Performance of the optimal resource allocation solution for a different workload ($W$) at $|\mathcal{K}_C|=5$. (a) Total computing efficiency gain. (b) Constraint value.}\label{fig:KKT_workload}
\end{figure}

In the first set of simulations, the performance of the optimal resource allocation solution is evaluated for $|\mathcal{K}_C| \in\{2,3,4,5,6\} $, with $W=6$. 
The transmitter-receiver distances of all the cooperative CAV pairs are set as $20$ $m$.  
For the constant workload, the initial binary-search interval, $[\nu_L, \nu_R]$, is the same for all the $|\mathcal{K}_C|$ values. 
Fig.~\ref{fig:KKT_CAVnum} shows the variations of total computing efficiency gain $\hat{G}^\ast$ and constraint value $h( \hat{\boldsymbol{f}}^\ast)$ during the binary-search of candidate optimal dual variable $\hat{\nu}^\ast$, for each $|\mathcal{K}_C|$.
The asymptotically optimal total computing efficiency gain and constraint value, $G^\ast$ and $h( \boldsymbol{f}^\ast)$, obtained at an asymptotically optimal dual variable, $\nu^\ast$, are represented by a red dot for each $|\mathcal{K}_C|$ value. 
As $|\mathcal{K}_C|$ increases, the radio resources are shared among more cooperative CAV pairs for feature data transmission, and each cooperative CAV pair should increase the CPU frequency to compensate for the lower average transmission rate, to achieve delay satisfaction. 
Accordingly, as $|\mathcal{K}_C|$ increases, the asymptotically optimal CPU frequency allocation variables, $\boldsymbol{f}^\ast$, are larger, corresponding to a larger $\nu^\ast$ value. 
Although the total reduced amount of computing demand increases in proportion to $|\mathcal{K}_C|$, the total computing efficiency gain gradually decreases as the CPU frequency further increases, leading to a first-increasing-then-decreasing trend of $G^\ast$, as shown in Fig.~\ref{fig:KKT_CAVnum}(a). 
At the asymptotically optimal points, constraint value $h( \boldsymbol{f}^\ast)$ approaches zero, as shown in Fig.~\ref{fig:KKT_CAVnum}(b).

In the second set of simulations, the performance of the optimal resource allocation solution is evaluated for $|\mathcal{K}_C|=5$, with shared workload $W\in\{4,5,6,7,8\}$, as shown in Fig.~\ref{fig:KKT_workload}. 
The transmitter-receiver distances of the $5$ cooperative CAV pairs are set as $[20.4, 16.5, 11.4, 29.7, 28.3]$ $m$, respectively. 
For a heavier workload, there is a right shift for the initial binary-search interval, $[\nu_L, \nu_R]$, as both $\nu_L$ and $\nu_R$ have a larger value. 
When $W$ increases, the total computing demand reduction increases proportionally, while the per-object delay budget, $\frac{\Delta}{W}$, decreases inverse proportionally. 
With a constant radio spectrum bandwidth, the CPU frequency should be increased at each CAV pair to satisfy the more stringent delay requirement as $W$ increases. 
Hence, in Fig.~\ref{fig:KKT_workload}(a), we observe a first-increasing-then-decreasing trend for $G^\ast$, due to a trade-off between computing demand and CPU frequency. 
Fig.~\ref{fig:KKT_workload}(b) shows that the constraint value, $h( \boldsymbol{f}^\ast)$, approaches zero at the asymptotically optimal points. 
Fig.~\ref{fig:KKT_CAVnum} and Fig.~\ref{fig:KKT_workload} demonstrate that, with a limited amount of radio resources, it is necessary to select the best subset of CAV pairs for cooperation while taking the shared workload into account, to improve the total computing efficiency gain.

% double
\begin{figure}
    \centering
    \begin{subfigure}[b]{0.38\textwidth}
        \includegraphics[width=1\linewidth]{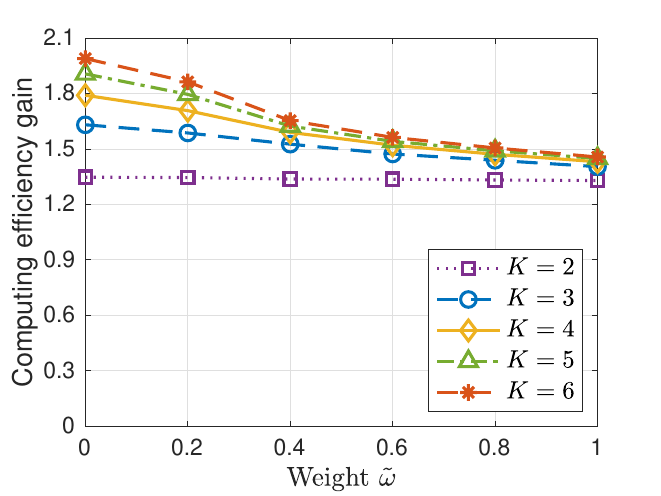} 
        \caption{}\label{fig:BF_weight_Gain}
    \end{subfigure}
    % ~ 
    \begin{subfigure}[b]{0.38\textwidth}
    \includegraphics[width=1\linewidth]{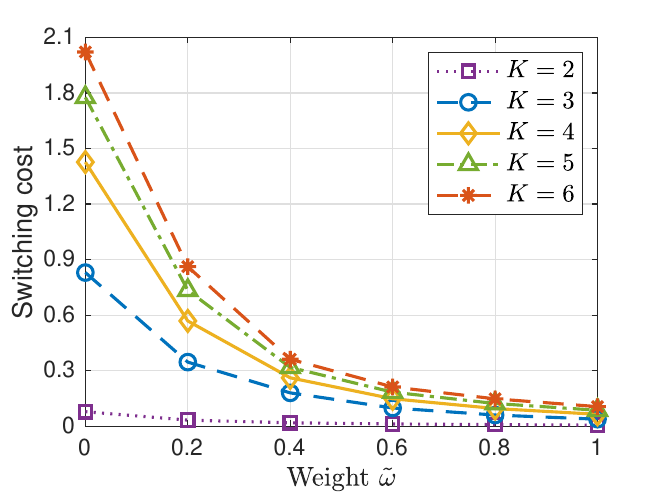} 
    \caption{}\label{fig:BF_weight_Cost}
    \end{subfigure}
    \caption{An illustration of gain-cost trade-off for different $\tilde{\omega}$ values. (a) Computing efficiency gain. (b) Switching cost.}\label{fig:BF_weight}
\end{figure}

Before the performance evaluation of the MADDPG algorithm for adaptive CAV cooperation, we examine the impact of weight $\tilde{\omega}$ in reward function~\eqref{eq-sys-reward} on the trade-off between computing efficiency gain and switching cost. We set $\tilde{\omega}\in\{0,0.2,0.4,0.6,0.8,1\}$. 
For each $\tilde{\omega}$ value, a group of experiments are performed for a different number of CAV pairs ($K$). In each experiment, a brute-force search is conducted among all possible CAV cooperation decisions for a maximum instantaneous reward in each time slot of $2000$ episodes. Fig.~\ref{fig:BF_weight} shows the slot-average performance in terms of computing efficiency gain and switching cost for different $K$ values as $\tilde{\omega}$ increases. As a larger $\tilde{\omega}$ value puts more emphasis on minimizing the switching cost, we observe a decreasing trend for both gain and cost with the increase of $\tilde{\omega}$. We also observe higher gain and cost for more CAV pairs at a given $\tilde{\omega}$ value, which is to be discussed later. The weight, $\tilde{\omega}$, can be selected according to the desired gain-cost trade-off. In the following, $\tilde{\omega}$ is set to $0.4$, as we observe that the cost is reduced by more than $80\%$ at a gain loss of less than $20\%$ in comparison with that achieved at $\tilde{\omega}=0$, for $K=6$.    

% double
\begin{figure}
    \centering
    \begin{subfigure}[b]{0.38\textwidth}
      \includegraphics[width=1\linewidth]{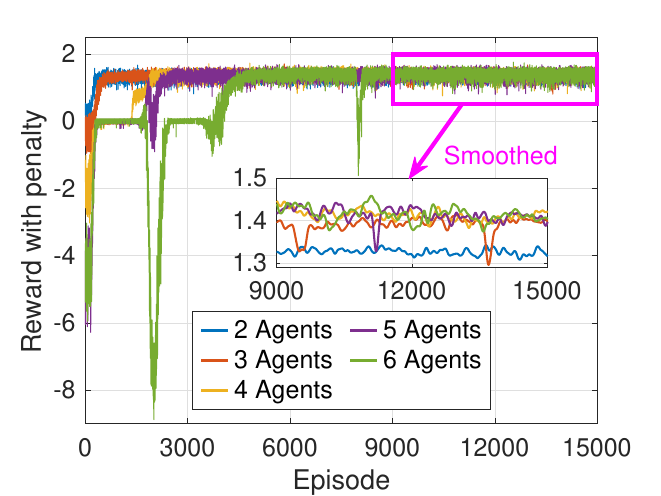} 
      \caption{}\label{fig:Convergence_reward_wi_penalty_diff_agents}
    \end{subfigure}
    ~~~
    \begin{subfigure}[b]{0.38\textwidth}
  \includegraphics[width=1\linewidth]{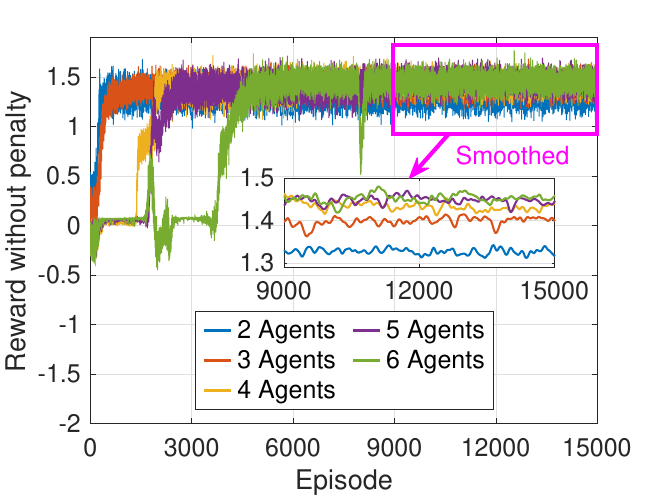} 
  \caption{}\label{fig:Convergence_reward_wo_penalty_diff_agents}
    \end{subfigure}
  \caption{Convergence of the reward during the training process. (a) With penalty. (b) Without penalty.}\label{fig:Convergence_rewards_diff_agents}
\end{figure}

Next, we evaluate the convergence of the MADDPG algorithm, in terms of both reward and training loss, for a different number of CAV pairs. 
Fig.~\ref{fig:Convergence_rewards_diff_agents} shows the convergence of the average reward per learning step over $15000$ episodes, for a different agent number ($K$) from $2$ to $6$. 
In a practical implementation, when there is a negative penalty in the reward due to infeasible resource allocation, we can refine the joint action and let all the CAV pairs work in the default SP mode, which gives a zero computing efficiency gain and a refined reward without penalty. 
The original reward with penalty and the original action are used during training to guide the learning agents towards a least penalty after convergence, while the reward without penalty is the true reward for the CAV pairs when the action refinement is enabled. 
Fig.~\ref{fig:Convergence_rewards_diff_agents} shows the average rewards with and without penalty during the training process. 
As $K$ increases, both rewards converge in a slower speed, indicated by a later increase to a converged value interval. 
By comparing the smoothed rewards with and without penalty after convergence, we see that the penalty has been effectively suppressed, indicated by the limited large negative glitches in the reward with penalty. It implies that the learning agents have collaboratively learn the adaptive CAV cooperation decisions that are feasible for resource allocation under the network dynamics. 
Moreover, we observe that more learning agents tend to improve the reward, as to be discussed.

Fig.~\ref{fig:Convergence_loss_diff_agents} shows the per-agent average critic and actor loss during training for $K$ from $2$ to $6$. 
For the critic network, the input dimension grows linearly with $K$. Thus, it is more difficult to minimize the critic loss if there are more agents, leading to a higher average critic loss after convergence as $K$ increases, as shown in Fig.~\ref{fig:Convergence_loss_diff_agents}(a).
As the reward tends to increase for more agents, as shown in Fig.~\ref{fig:Convergence_rewards_diff_agents}, the average $Q$-value for the sampled mini-batch of experiences at each learning step increases for more agents. As the actor loss is the negative average $Q$-value, there is a lower average actor loss after convergence as $K$ increases, as shown in Fig.~\ref{fig:Convergence_loss_diff_agents}(b).

% double
\begin{figure}
    \centering
    \begin{subfigure}[b]{0.38\textwidth}
      \includegraphics[width=1\linewidth]{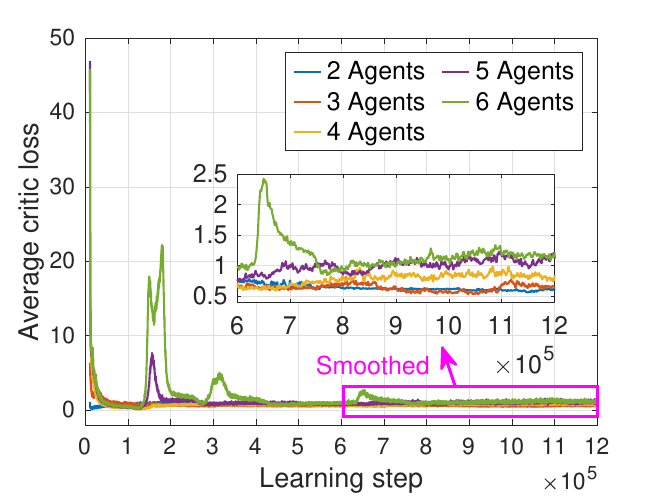} 
      \caption{}\label{fig:CConvergence_critic_loss_diff_agents}
    \end{subfigure}
    ~~~
    \begin{subfigure}[b]{0.38\textwidth}
  \includegraphics[width=1\linewidth]{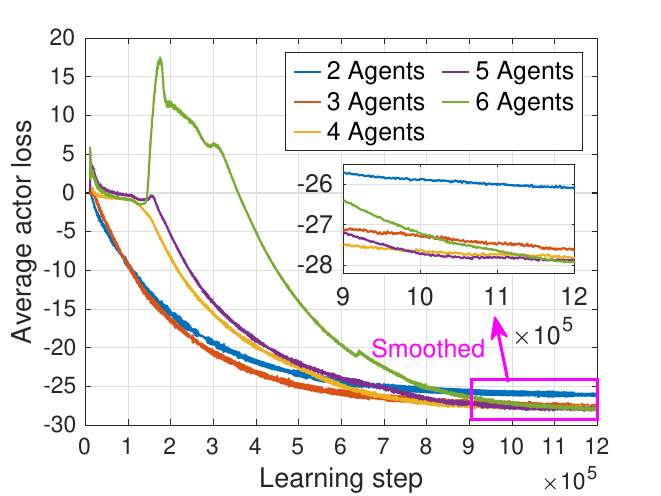} 
  \caption{}\label{fig:Convergence_actor_loss_diff_agents}
    \end{subfigure}
  \caption{Convergence of the critic loss and actor loss during the training process. (a) Average critic loss. (b) Average actor loss.}\label{fig:Convergence_loss_diff_agents}
\end{figure}

To evaluate the effectiveness of the MADDPG algorithm in improving the computing efficiency gain and reducing the switching cost, we compare the performance between the trained MADDPG algorithm and three benchmark algorithms, for a different number of CAV pairs ($K$). 
The first benchmark is a random CAV cooperation scheme, in which each CAV pair switches between SP and CP modes at random in each time slot. 
Action refinement is enabled.  
In the second benchmark, all CAV pairs always cooperate if there exists a feasible resource allocation solution among them. Otherwise, action refinement is triggered to let all CAV pairs work in the default SP mode. Hence, the solution switches between all CAV pairs working in the CP mode and all CAV pairs working in the SP mode, referred to as ``all CP mode'' and ``all SP mode'' respectively.  
In the third benchmark, we conduct a step-wise brute-force search among all candidate joint CAV cooperation decisions in each time slot, for a maximum instantaneous reward. 
The time complexity for the brute-force benchmark is $2^K$ times of that for a trained MADDPG algorithm, for solving $2^K$ resource allocation subproblems given each candidate joint CAV cooperation decision. 
By using a trained MADDPG algorithm, only one resource allocation subproblem is solved given one learned joint action. 
For performance comparison, the $25\%$, $50\%$, and $75\%$ percentiles of the slot-average total computing efficiency gain, slot-average total switching cost, and slot-average reward are evaluated for each $K$ value using the four algorithms, as shown in Fig.~\ref{fig:Gain_Cost_Reward_comp_diff_agent}.

% double
\begin{figure}
    \centering
    \begin{subfigure}[b]{0.46\textwidth}
      \includegraphics[width=1\linewidth]{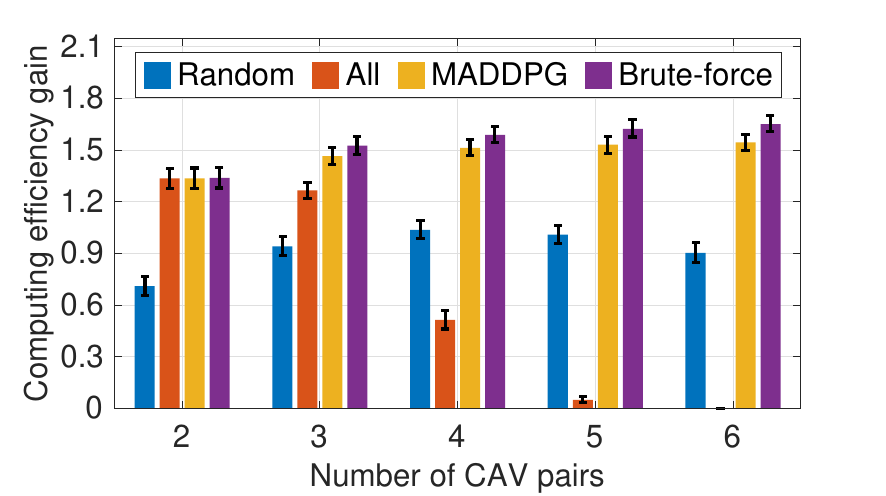} 
      \caption{}\label{fig:Gain_comp_diff_agent}
    \end{subfigure}
    % ~~
    \begin{subfigure}[b]{0.46\textwidth}
  \includegraphics[width=1\linewidth]{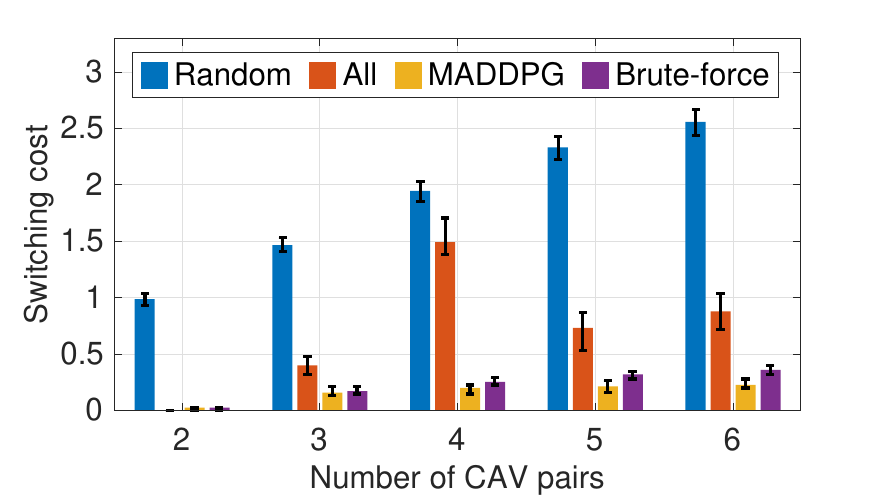} 
  \caption{}\label{fig:Cost_comp_diff_agent}
    \end{subfigure}
    % ~~ 
    \begin{subfigure}[b]{0.46\textwidth}
  \includegraphics[width=1\linewidth]{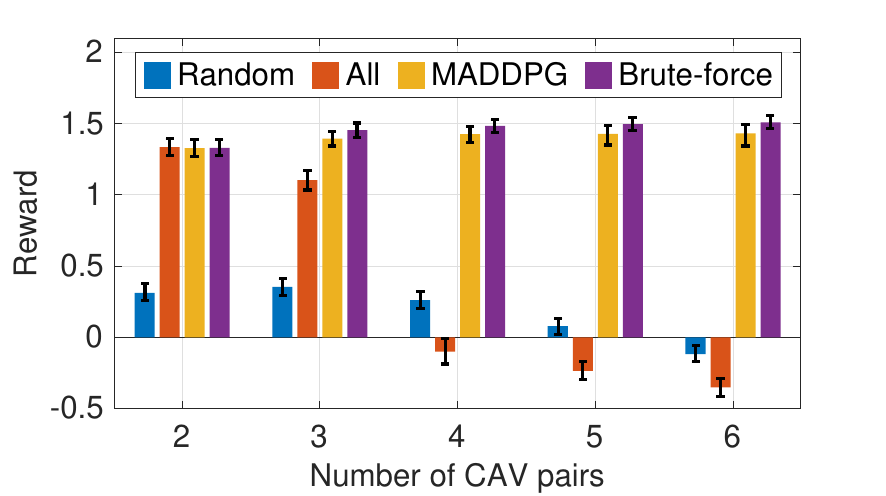} 
  \caption{}\label{fig:Reward_comp_diff_agent}
    \end{subfigure}
    \caption{Performance comparison between the proposed and benchmark algorithms. (a) Computing efficiency gain. (b) Switching cost. (c) Reward.}\label{fig:Gain_Cost_Reward_comp_diff_agent}
\end{figure}

Fig.~\ref{fig:Gain_Cost_Reward_comp_diff_agent}(a) demonstrates an increasing trend of the computing efficiency gain as the number of CAV pairs increases for both the MADDPG algorithm and the brute-force benchmark, both of which significantly outperform the other two benchmarks.
Fig.~\ref{fig:KKT_CAVnum}(a) demonstrates the first-increasing-then-decreasing computing efficiency gain when more CAV pairs cooperate. 
Accordingly, for the random benchmark in which the average number of selected cooperative CAV pairs increases linearly with $K$, we observe such a trend in the computing efficiency gain as $K$ increases, with a turning point at $K=4$. 
The second benchmark achieves a gradually degraded computing efficiency gain when $K$ increases, until obtaining a zero gain at $K=6$. For a larger $K$ value, the chance for infeasible resource allocation among all CAV pairs increases, leading to more frequent action refinement to the ``all SP mode'' with zero gain. 
Nevertheless, in both the MADDPG algorithm and the brute-force benchmark, as there are more candidate CAV pairs in the system, there is a higher flexibility in selecting the best subset of CAV pairs for cooperation, by considering their shared workloads and transmitter-receiver distances. This flexibility contributes to a further increase in the computing efficiency gain as $K$ further increases from $4$ to $6$. 
We observe a decreasing speed in the increase of computing efficiency gain as $K$ further increases, as the gain from the higher flexibility gradually saturates. 
Especially, it is more difficult for the MADDPG algorithm to find the optimal solution 
in a distributed fashion as the agent number, $K$, increases, leading to a slightly larger sub-optimality gap from the brute-force benchmark. 

From Fig.~\ref{fig:Gain_Cost_Reward_comp_diff_agent}(b), we observe an almost linear increasing switching cost for the random benchmark, while both the MADDPG algorithm and the brute-force benchmark can significantly reduce the switching cost. 
Especially, as the switching cost has time correlation, the MADDPG algorithm which can minimize the total switching cost in the long run achieves a lower average switching cost in comparison with the step-wise brute-force benchmark which does not take the future states into account.  
For the second benchmark, as $K$ increases, the dominant solution gradually changes from the ``all CP mode'' to the ``all SP mode'', due to a higher action refinement probability. Accordingly, the highest switching cost is obtained at a medium $K$ value. 
For both the MADDPG algorithm and the brute-force benchmark, 
due to a higher flexibility in cooperative CAV pair selection as $K$ increases, there are more changes in the selection decision, leading to an increasing switching cost as $K$ increases. 
The reward, which is a linear combination of the computing efficiency gain and the switching cost, is shown in Fig.~\ref{fig:Gain_Cost_Reward_comp_diff_agent} (c) for reference.

\section{Conclusion}
\label{Conclusion}

In this paper, we develop an adaptive cooperative perception framework for CAVs in a moving mixed-traffic vehicle cluster, while considering the dynamic shared workloads and channel conditions due to vehicle mobility, dynamic radio resource availability, and intermittent RSU coverage. 
A model-assisted multi-agent reinforcement learning solution is developed, which integrates learning-based adaptive CAV cooperation decision over time with model-based resource allocation decision in each time slot. 
Simulation results demonstrate the necessity for dynamically activating the cooperative perception among CAV pairs to improve the total computing efficiency gain. 
The effectiveness of the model-assisted MADDPG algorithm is verified, in improving the total computing efficiency gain under a limited switching cost. 
In future works, we will explore the generalization of the learning model for adaptive CAV cooperation, while considering a varying vehicle cluster size and a non-stationary network environment.

\section*{Appendix: Proof of Theorem 1}
\label{sec: Appendix}
\setcounter{equation}{0}
\renewcommand{\theequation}{A\arabic{equation}}

\begin{Lem}
\emph{Constraint \eqref{eq-constraint-func-inequality} must be active for an optimal solution to problem $\mathbf{P}_2$.} 
\end{Lem}

\begin{proof}
Assume there is an optimal solution to $\mathbf{P}_2$ which achieves ``$<$'' in \eqref{eq-constraint-func-inequality}.  Then, there must be another feasible solution achieving ``$=$'' in \eqref{eq-constraint-func-inequality} by decreasing the CPU frequency for some CAV pairs, thus further decreasing the objective value in \eqref{eq-obj} without violating constraint \eqref{eq-CPU-upper-bound}, as the objective function is an increasing function of $\boldsymbol{f}_C(n)$ while constraint function $h( \boldsymbol{f} )$ is a decreasing function of $\boldsymbol{f}_C(n)$. 
Therefore, the assumption must be false, and constraint \eqref{eq-constraint-func-inequality} must be active for an optimal solution. Lemma 1 is proved. 
\end{proof}

For problem $\mathbf{P}_2$, the objective function is a second-order function of decision variable $f_k(n)\in\boldsymbol{f}_C(n)$ with coefficient $W_k(n)>0$, which is convex.  
Constraint \eqref{eq-CPU-upper-bound} is linear.  
For constraint \eqref{eq-constraint-func-inequality}, the second-order derivative of constraint function $h( \boldsymbol{f} )$ with respect to $f_k(n)$, is given by
\begin{align} 
	\frac{\partial h^2( \boldsymbol{f} )}{\partial \left(f_k(n)\right)^2} = \frac{2 c_k b_k \hat{\delta}}{\left( b_k f_k(n) - \hat{\delta}\right)^3}.
\end{align}
As $f_k(n) >  \frac{\hat{\delta}}{b_k}$, we have $\partial h^2( \boldsymbol{f} ) / \partial \left(f_k(n)\right)^2 >0$, thus $h( \boldsymbol{f} )$ is convex.
Therefore, $\mathbf{P}_2$ is convex. 

For a convex problem, strong duality holds if Slater's condition (or a weaker form) is satisfied, which requires the nonlinear inequality constraints to be strictly feasible~\cite{boyd2004convex}. 
For problem $\mathbf{P}_2$, there is only one nonlinear inequality constraint in \eqref{eq-constraint-func-inequality}. Suppose the optimal solution is denoted as $\boldsymbol{f}_C^\ast(n)$ if the problem is feasible. 
Based on Lemma 1, we have $h( \boldsymbol{f}^\ast ) = 0$. Then, as long as $h( \boldsymbol{f}^0 ) < 0$, there must exist at least one strictly feasible non-optimal solution $\boldsymbol{f}_C^\diamond(n) \succeq \boldsymbol{f}_C^\ast(n)$ which gives a larger objective value while satisfying $\boldsymbol{f}_C^\diamond(n) \preceq \boldsymbol{f}_C^0(n)$ and $h( \boldsymbol{f}^\diamond ) < 0$.
Therefore, strong duality holds for $\mathbf{P}_2$ if $h( \boldsymbol{f}^0 ) < 0$. Theorem 1 is proved. 
Note that, if $h( \boldsymbol{f}^0 ) = 0$, we can directly obtain the optimal solution as $\boldsymbol{f}_C^\ast(n) = \boldsymbol{f}_C^0(n)$; if $h( \boldsymbol{f}^0 ) > 0$, the problem is infeasible.

\bibliographystyle{IEEEtran}
\bibliography{ref_final}

% Generated by IEEEtran.bst, version: 1.14 (2015/08/26)
\begin{thebibliography}{10}
\providecommand{\url}[1]{#1}
\csname url@samestyle\endcsname
\providecommand{\newblock}{\relax}
\providecommand{\bibinfo}[2]{#2}
\providecommand{\BIBentrySTDinterwordspacing}{\spaceskip=0pt\relax}
\providecommand{\BIBentryALTinterwordstretchfactor}{4}
\providecommand{\BIBentryALTinterwordspacing}{\spaceskip=\fontdimen2\font plus
\BIBentryALTinterwordstretchfactor\fontdimen3\font minus
  \fontdimen4\font\relax}
\providecommand{\BIBforeignlanguage}[2]{{%
\expandafter\ifx\csname l@#1\endcsname\relax
\typeout{** WARNING: IEEEtran.bst: No hyphenation pattern has been}%
\typeout{** loaded for the language `#1'. Using the pattern for}%
\typeout{** the default language instead.}%
\else
\language=\csname l@#1\endcsname
\fi
#2}}
\providecommand{\BIBdecl}{\relax}
\BIBdecl

\bibitem{zhuang2019sdn}
W.~Zhuang, Q.~Ye, F.~Lyu, N.~Cheng, and J.~Ren, ``{SDN/NFV}-empowered future
  {IoV} with enhanced communication, computing, and caching,'' \emph{Proc.
  IEEE}, vol. 108, no.~2, pp. 274--291, 2019.

\bibitem{shen2021holistic}
X.~Shen, J.~Gao, W.~Wu, M.~Li, C.~Zhou, and W.~Zhuang, ``Holistic network
  virtualization and pervasive network intelligence for {6G},'' \emph{IEEE
  Commun. Surv. Tutor.}, vol.~24, no.~1, pp. 1--30, 2021.

\bibitem{hui2022collaboration}
Y.~Hui, X.~Ma, Z.~Su, N.~Cheng, Z.~Yin, T.~H. Luan, and Y.~Chen,
  ``Collaboration as a service: Digital-twin-enabled collaborative and
  distributed autonomous driving,'' \emph{IEEE Internet Things J.}, vol.~9,
  no.~19, pp. 18\,607--18\,619, 2022.

\bibitem{wang2018networking}
J.~Wang, J.~Liu, and N.~Kato, ``Networking and communications in autonomous
  driving: A survey,'' \emph{IEEE Commun. Surv. Tutor.}, vol.~21, no.~2, pp.
  1243--1274, 2018.

\bibitem{zhang2019mobile}
J.~Zhang and K.~B. Letaief, ``Mobile edge intelligence and computing for the
  {Internet} of vehicles,'' \emph{Proc. IEEE}, vol. 108, no.~2, pp. 246--261,
  2019.

\bibitem{zheng2022confidence}
X.~Zheng, S.~Li, Y.~Li, D.~Duan, L.~Yang, and X.~Cheng, ``Confidence evaluation
  for machine learning schemes in vehicular sensor networks,'' \emph{IEEE
  Trans. Wirel. Commun.}, vol.~22, no.~4, pp. 2833--2846, 2023.

\bibitem{jia2022online}
Y.~Jia, R.~Mao, Y.~Sun, S.~Zhou, and Z.~Niu, ``Online {V2X} scheduling for
  raw-level cooperative perception,'' in \emph{Proc. IEEE ICC}, 2022, pp.
  309--314.

\bibitem{jia2023mass}
------, ``Mass: Mobility-aware sensor scheduling of cooperative perception for
  connected automated driving,'' \emph{IEEE Trans. Veh. Technol.}, vol.~72,
  no.~11, pp. 14\,962--14\,977, 2023.

\bibitem{abdel2021vehicular}
M.~K. Abdel-Aziz, C.~Perfecto, S.~Samarakoon, M.~Bennis, and W.~Saad,
  ``Vehicular cooperative perception through action branching and federated
  reinforcement learning,'' \emph{IEEE Trans. Commun.}, vol.~70, no.~2, pp.
  891--903, 2022.

\bibitem{xiao2022perception}
Z.~Xiao, J.~Shu, H.~Jiang, G.~Min, H.~Chen, and Z.~Han, ``Perception task
  offloading with collaborative computation for autonomous driving,''
  \emph{IEEE J. Sel. Areas Commun.}, vol.~41, no.~2, pp. 457--473, 2023.

\bibitem{sun2022user}
Y.~Sun, J.~Xu, and S.~Cui, ``User association and resource allocation for
  {MEC}-enabled {IoT} networks,'' \emph{IEEE Trans. Wirel. Commun.}, vol.~21,
  no.~10, pp. 8051--8062, 2022.

\bibitem{lin2022low}
J.~Lin, P.~Yang, N.~Zhang, F.~Lyu, X.~Chen, and L.~Yu, ``Low-latency edge video
  analytics for on-road perception of autonomous ground vehicles,'' \emph{IEEE
  Trans. Industr. Inform.}, vol.~19, no.~2, pp. 1512--1523, 2022.

\bibitem{zhang2021emp}
X.~Zhang, A.~Zhang, J.~Sun, X.~Zhu, Y.~E. Guo, F.~Qian, and Z.~M. Mao, ``{EMP}:
  Edge-assisted multi-vehicle perception,'' in \emph{Proc. 27th Annual
  International Conf. Mobile Computing and Networking}, 2021, pp. 545--558.

\bibitem{chen2019cooper}
Q.~Chen, S.~Tang, Q.~Yang, and S.~Fu, ``Cooper: {C}ooperative perception for
  connected autonomous vehicles based on {3D} point clouds,'' in \emph{Proc.
  IEEE Int. Conf. Distrib. Comput. Syst. (ICDCS)}, 2019, pp. 514--524.

\bibitem{qiu2021autocast}
H.~Qiu, P.~Huang, N.~Asavisanu, X.~Liu, K.~Psounis, and R.~Govindan,
  ``Autocast: Scalable infrastructure-less cooperative perception for
  distributed collaborative driving,'' in \emph{Proc. ACM Int. Conf. Mobile
  Syst., Appl. and Services (MobiSys)}, 2021, pp. 128--141.

\bibitem{chen2019f}
Q.~Chen, X.~Ma, S.~Tang, J.~Guo, Q.~Yang, and S.~Fu, ``F-cooper: {F}eature
  based cooperative perception for autonomous vehicle edge computing system
  using {3D} point clouds,'' in \emph{Proc. ACM/IEEE Symp. Edge Comput. (SEC)},
  2019, pp. 88--100.

\bibitem{wang2020v2vnet}
T.-H. Wang, S.~Manivasagam, M.~Liang, B.~Yang, W.~Zeng, and R.~Urtasun,
  ``{V2VNeT}: Vehicle-to-vehicle communication for joint perception and
  prediction,'' in \emph{Proc. Eur. Conf. Comput. Vision (ECCV)}, 2020, pp.
  605--621.

\bibitem{wu2020dynamic}
W.~Wu, N.~Chen, C.~Zhou, M.~Li, X.~Shen, W.~Zhuang, and X.~Li, ``Dynamic {RAN}
  slicing for service-oriented vehicular networks via constrained learning,''
  \emph{IEEE J. Sel. Areas Commun.}, vol.~39, no.~7, pp. 2076--2089, 2021.

\bibitem{zhang2021elf}
W.~Zhang, Z.~He, L.~Liu, Z.~Jia, Y.~Liu, M.~Gruteser, D.~Raychaudhuri, and
  Y.~Zhang, ``{Elf}: accelerate high-resolution mobile deep vision with
  content-aware parallel offloading,'' in \emph{Proc. ACM MobiCom}, 2021, pp.
  201--214.

\bibitem{wang2022vabus}
H.~Wang, Q.~Li, H.~Sun, Z.~Chen, Y.~Hao, J.~Peng, Z.~Yuan, J.~Fu, and Y.~Jiang,
  ``Vabus: Edge-cloud real-time video analytics via background understanding
  and subtraction,'' \emph{IEEE J. Select. Areas Commun.}, vol.~41, no.~1, pp.
  90--106, 2023.

\bibitem{yang2022flexpatch}
K.~Yang, J.~Yi, K.~Lee, and Y.~Lee, ``{FlexPatch}: Fast and accurate object
  detection for on-device high-resolution live video analytics,'' in \emph{IEEE
  Proc. INFOCOM}, 2022, pp. 1898--1907.

\bibitem{teerapittayanon2017distributed}
S.~Teerapittayanon, B.~McDanel, and H.-T. Kung, ``Distributed deep neural
  networks over the cloud, the edge and end devices,'' in \emph{Proc. IEEE
  ICDCS}, 2017, pp. 328--339.

\bibitem{10137743}
K.~Qu, W.~Zhuang, W.~Wu, M.~Li, X.~Shen, X.~Li, and W.~Shi, ``Stochastic
  cumulative {DNN} inference with {RL}-aided adaptive {IoT} device-edge
  collaboration,'' \emph{IEEE Internet Things J.}, vol.~10, no.~20, pp.
  18\,000--18\,015, 2023.

\bibitem{li2019edge}
E.~Li, L.~Zeng, Z.~Zhou, and X.~Chen, ``Edge {AI}: On-demand accelerating deep
  neural network inference via edge computing,'' \emph{IEEE Trans. Wirel.
  Commun.}, vol.~19, no.~1, pp. 447--457, 2020.

\bibitem{liu2023resource}
Z.~Liu, Q.~Lan, and K.~Huang, ``Resource allocation for multiuser edge
  inference with batching and early exiting,'' \emph{IEEE J. Select. Areas
  Commun.}, vol.~41, no.~4, pp. 1186--1200, 2023.

\bibitem{wang2023real}
H.~Wang, H.~Bao, L.~Zeng, K.~Luo, and X.~Chen, ``Real-time high-resolution
  pedestrian detection in crowded scenes via parallel edge offloading,'' in
  \emph{Proc. IEEE ICC}, 2023.

\bibitem{qu2020dynamic1}
K.~Qu, W.~Zhuang, Q.~Ye, X.~Shen, X.~Li, and J.~Rao, ``Dynamic flow migration
  for embedded services in {SDN/NFV}-enabled {5G} core networks,'' \emph{IEEE
  Trans. Commun.}, vol.~68, no.~4, pp. 2394--2408, 2020.

\bibitem{boyd2004convex}
S.~Boyd, S.~P. Boyd, and L.~Vandenberghe, \emph{Convex optimization}.\hskip 1em
  plus 0.5em minus 0.4em\relax Cambridge university press, 2004.

\bibitem{lowe2017multi}
R.~Lowe, Y.~I. Wu, A.~Tamar, J.~Harb, O.~Pieter~Abbeel, and I.~Mordatch,
  ``Multi-agent actor-critic for mixed cooperative-competitive environments,''
  in \emph{Proc. 30th Adv. Neural Inf. Process. Syst. (NeurIPS)}, 2017, pp.
  6379--6390.

\bibitem{tian2021multiagent}
J.~Tian, Q.~Liu, H.~Zhang, and D.~Wu, ``Multiagent
  deep-reinforcement-learning-based resource allocation for heterogeneous {QoS}
  guarantees for vehicular networks,'' \emph{IEEE Internet Things J.}, vol.~9,
  no.~3, pp. 1683--1695, 2022.

\bibitem{ye2021joint}
Q.~Ye, W.~Shi, K.~Qu, H.~He, W.~Zhuang, and X.~Shen, ``Joint {RAN} slicing and
  computation offloading for autonomous vehicular networks: {A}
  learning-assisted hierarchical approach,'' \emph{IEEE Open J. Veh. Technol.},
  vol.~2, pp. 272--288, 2021.

\bibitem{jang2016categorical}
E.~Jang, S.~Gu, and B.~Poole, ``Categorical reparameterization with
  gumbel-softmax,'' \emph{Proc. ICLR}, 2017.

\bibitem{SUMO}
``{Simulation of Urban MObility (SUMO) 1.16.0},''
  \url{https://www.eclipse.org/sumo/}, 2023, [Online; accessed 8-January-2024].

\bibitem{3gpp.37.885}
3GPP, ``{Study on evaluation methodology of new Vehicle-to-Everything (V2X) use
  cases for LTE and NR},'' {3rd Generation Partnership Project (3GPP)},
  Technical Specification (TS) 37.885, 2019, version 15.3.0.

\end{thebibliography}

\end{document}